\documentclass{article}
\usepackage[utf8]{inputenc}

\usepackage{booktabs} % For formal tables
\usepackage[ruled,vlined]{algorithm2e}
\usepackage{amsfonts}
\usepackage{amsmath}
\usepackage{amssymb}
\usepackage{amsthm}
\usepackage{authblk}
\usepackage{comment}
\usepackage{float}
\usepackage{fullpage}
\usepackage{graphicx}
\usepackage{hyperref}
\usepackage{mathtools}
\usepackage{thm-restate}

\theoremstyle{definition}
\newtheorem{assumption}{Assumption}
  
\newtheorem{definition}{Definition}
\newtheorem{lemma}{Lemma}

\newtheorem{observation}{Observation}
\newtheorem{theorem}{Theorem}
\newtheorem*{theorem*}{Theorem}
\newtheorem*{remark}{Remark}

\hypersetup{colorlinks=true,linkcolor=blue,citecolor=blue,urlcolor=blue}
\usepackage[capitalize]{cleveref}
\usepackage[multiple]{footmisc}

\newcommand{\LL}{{\mathcal{L}}}

\newcommand*\samethanks[1][\value{footnote}]{\footnotemark[#1]}

\title{Beyond Pigouvian Taxes: A Worst Case Analysis}

\author[1]{Moshe Babaioff}
\author[2]{Ruty Mundel\thanks{Supported by the European Research Council (ERC) under the European Union’s Horizon 2020 research and innovation programme (grant agreement No 740282).}\thanks{Supported by the Israeli Smart Transportation Research Center (ISTRC).}}
\author[2]{Noam Nisan\samethanks[1]}
\affil[1]{Microsoft Research}
\affil[2]{Hebrew University of Jerusalem, Israel}

\begin{document}
	
	\maketitle
	\begin{abstract}
		In the early 20\textsuperscript{th} century, Pigou observed that imposing a marginal cost tax on the usage of a public good induces a socially efficient level of use as an equilibrium. Unfortunately, such a ``Pigouvian" tax may also induce other, socially inefficient, equilibria. We observe that this social inefficiency may be unbounded, and study whether alternative tax structures may lead to milder losses in the worst case, i.e. to a lower price of anarchy. We show that no tax structure leads to bounded losses in the worst case. However, we do find a tax scheme that has a lower price of anarchy than the Pigouvian tax, obtaining tight lower and upper bounds in terms of a crucial parameter that we identify. We generalize our results to various scenarios that each offers an alternative to the use of a public road by private cars, such as ride sharing, or using a bus or a train.	
	\end{abstract}
	\section{Introduction}
	\label{sec:intro}
	
	This paper studies the design of taxes intended to overcome the ``tragedy of the commons"  in the use of a shared public resource. We consider a situation with a public good and a population of users where each of them may choose to use either the public good or to use a more costly alternative instead. Users of the public good ``congest” it, causing a negative externality to all others so the social planner wishes to reduce the use of the public good {to the socially optimal level,} by levying taxes on such use. Examples of this scenario abound, e.g.:
	
	\begin{itemize}
		\item Road tolls. Commuters may either use the road by driving their car to work or take public transportation. There is an inconvenience for taking public transportation, but driving a car increases the congestion on the road leading to increased commute times for everyone. Tolls on the road may incentivize taking public transportation.
		\item Carbon taxes. People that use carbon-based energy sources may instead opt to use more expensive renewable energy sources (e.g. an electric car vs. a petrol based car). The carbon-based energy sources have an externality in terms of pollution and global warming. Carbon taxes incentivize a switch to renewable energy sources.
	\end{itemize}
	
	\subsection{Public-Good Congestion Games}
	Here is a very basic model for these situations. We capture the demand for the public good by the \emph{individual cost function} $\alpha : [0,1]  \rightarrow \mathbb{R}_{\ge0}$, where $\alpha(q)$ is the price at which fraction $q$ of the population chooses to use the public good.\footnote{The individual cost function is also the inverse function of the demand function.
		For every $q\in [0,1]$,  $q$ fraction  of the population  have disutility at least $\alpha(q)$ for not using the public good (and using the alternative instead).}
	We capture the negative externality imposed by the use of the public good by a non-decreasing \emph{externality function} $l: [0,1]  \rightarrow \mathbb{R}_{\ge0}$, where $l(q)$
	is the negative externality	when fraction $q$ of the population uses the public good. We assume that this externality is borne by every member of the population\footnote{We also study a variant where only those that use the public good incur the cost of $l(q)$, see Section \ref{sec:train}.}. 
	{The \emph{Social Cost} when fraction $q$ of the population uses the public good is thus $SC(q)= l(q) + \int_q^1 \alpha(x)dx$.}
	In this model, if the public good is offered for free then everyone will use the public good so the total social cost will be $l(1)$, which may be tragically high relative to the socially optimal usage level, the one that minimizes this social cost.	Already in 1920, Pigou \cite{Pigou1920} suggested a taxation scheme that will result in an efficient equilibrium: tax each user of the public good the marginal externality that she imposes on society, i.e. the taxation function is given by $t(q)=l'(q)$, where $q$ is the fraction of the population using the public good (and $l'(\cdot)$ is the derivative of $l$ with respect to $q$). Using first order conditions to minimize the social cost $l(q) + \int_q^1 \alpha(x)dx$ we get that for an optimal $q$ we have that $\alpha(q)=l'(q)$ which is indeed obtained as the equilibrium with this Pigouvian tax. 
	
	The starting point of this paper notes that this analysis only guarantees the existence of an efficient equilibrium\footnote{Under appropriate continuity assumptions.}, but it ignores the possibility of other, non-efficient equilibria. It turns out that  existence of bad equilibria is indeed possible. Figure \ref{fig:IntroductionExample} shows an example where there are three equilibrium points with $\alpha(q)=l'(q)$, at least one of which has a tragically high social cost, while another has very low social cost. 
	Furthermore, as we demonstrate in Appendix \ref{appendix:simulations}, it is indeed possible that natural, best-reply, dynamics will lead to this socially-costly equilibrium.
	
	\begin{figure}[H]
		\centering
		\includegraphics[width=5in]{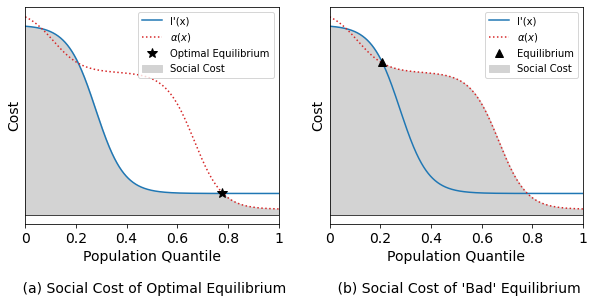}
		\caption{An example in which the Pigouvian tax induces both the optimal point as an equilibrium (a),  as well as a 'bad' equilibrium (b). For each  population quantile $q\in[0,1]$, the total social cost is given by the shaded area, combining the area under the externality derivative $l'$ (total externality) to its left, and the area under the individual cost function $\alpha$ to its right.}
		\label{fig:IntroductionExample}
	\end{figure}
	
	The example in Figure \ref{fig:IntroductionExample} demonstrates that there may exist non-optimal equilibrium points, but how bad can they be? In other words, what is the ``price of anarchy" in this setting? Taking the example of Figure \ref{fig:IntroductionExample} to its limits, we observe that the ratio between the  social cost of an equilibrium and the optimal social cost may be unbounded.
	\begin{observation}
		The price of anarchy of the Pigouvian tax is unbounded.
	\end{observation}
	
	Given this failure of the Pigouvian tax, the question that we ask is whether some other taxation scheme $t(\cdot)$ can guarantee a reasonably good social cost in all equilibria points. In other words, yield a better price of anarchy. Like the Pigouvian tax, the desired taxing scheme will charge a price according to the current load in the system. 
	Our main result is negative, showing no tax function yields a small price of anarchy. To formulate this we need some  definitions: 
	
	\vspace{0.5cm}
	\textbf{Definitions}
	\begin{itemize}
		\item For a fixed externality function $l:[0,1]\longrightarrow\mathbb{R}_{\ge0}$, a fixed tax function $t:[0,1]\longrightarrow\mathbb{R}_{\ge0}$  and a fixed individual cost function $\alpha:[0,1]\longrightarrow\mathbb{R}_{\ge0}$, the \emph{Price of Anarchy} $PoA(l,t,\alpha)$ is the social cost of the worst equilibrium\footnote{If no equilibrium exists then we define the price of anarchy to be $1$. An equilibrium always exists if $t$ is continuous. None of our results rely on non-existence of equilibrium.} divided by the cost of the social optimum, in the game induced by $l$, $t$, and $\alpha$.
		\item For a given fixed externality function $l:[0,1]\longrightarrow\mathbb{R}_{\ge0}$, the \emph{Taxed Price of Anarchy for $l$}, $TPoA(l)$, is the price of anarchy obtained for the worst individual cost curve $\alpha$ under the best taxation scheme $t(\cdot)$:
		$TPoA(l) = \inf_{t}\sup_{\alpha} PoA(l,t,\alpha)$.
		\item The \emph{Taxed Price of Anarchy of a family (set) of externality functions $\LL=\{l(\cdot)\}$}, is the taxed price of anarchy of the worst externality function in the family: $TPoA(\LL) = \sup_{l \in \LL} TPoA(l)$.
	\end{itemize}
	\vspace{0.5cm}
	
	Note that under the definition of $TPoA(l)$  the tax function $t(\cdot)=t_{l}(\cdot)$ may depend on the externality function $l(\cdot)$ as we assume it is public knowledge but may not depend on the individual cost function $\alpha(\cdot)$ which we assume is unknown to the planner and possibly varies over time.\footnote{It is not difficult to see that the problem becomes trivial if the demand is fixed and the tax function may depend on it.} This basic modeling choice, which follows Pigou, makes it impossible to analyze our problem by encoding it as a routing game, even when the number of agents is finite. This is so since this would require explicitly encoding $\alpha$ within the graph, while we assume $\alpha$ is unknown to the tax designer.
	
	Before presenting a full version of our main negative result, we state a qualitative version, showing that no tax function has a bounded price of anarchy:
	
	\vspace*{0.5cm}
	\textbf{Theorem (``Lower Bound", qualitative version)}: The taxed-price-of-anarchy of the family of all monotone non-decreasing externality functions is unbounded.
	\vspace*{0.5cm}
	
	Thus, the Pigouvian tax is not the only one susceptible to the problem of agents being trapped in a bad equilibrium given that tax, but rather, \emph{every} tax suffers from that problem when demand is adversarial.
	Given this negative result we look for conditions under which the taxed price of anarchy can be bounded using some carefully chosen tax scheme.
	
	In the classic special case of strictly convex externality functions (i.e. increasing $l'(\cdot)$) the Pigouvian tax ensures that there is only a single equilibrium point which is thus optimal and so the price of anarchy is 1.  Convexity of $l$ is, however, a strong assumption that may not always hold, e.g. when the externality function $l$ has a sigmoid-like shape in cases where the marginal externality saturates at a certain load. We extend the classic convex case to externality functions $l(\cdot)$ that have derivative $l'(\cdot)$ that is ``approximately monotone" and show that the taxed price of anarchy degrades linearly in the monotonicity approximation level. See Theorem \ref{thm:bus-gamma}.
	
	Our main results concern general (far from convex) externality functions for which we present tight upper and lower bounds for the taxed price of anarchy as a function of a crucial parameter which we identify. This parameter is the maximal ratio between derivative of $l$ at any two points, which we denote by $H$.
	
	\begin{definition}
		For a fixed constant $H>1$, let $\LL_H$ be the \emph{family of functions with bounded derivative ratio $H$:}
		the family of externality functions that each satisfies  $l'(q_1)/l'(q_2) \le H$ for every $0 \le q_1, q_2 \le 1$.
	\end{definition}
	
	We first prove an upper bound on the taxed price of anarchy of  $\LL_H$:
	
	\vspace*{0.5cm}
	\textbf{Theorem (``Upper Bound")}: For $\LL_H$,  the {family of functions with bounded derivative ratio $H>1$, it holds that the taxed price of anarchy of $\LL_H$ is at most $\sqrt H$, i.e. $TPoA(\LL_{H})\le\sqrt{H}$.
	\vspace{0.5cm}
		
	This theorem is obtained with a simple constant tax that does not depend on the load and is set to be the geometric mean of the maximum and minimum values of $l’$. In contrast we show in Appendix \ref{appendix:PigouBound} that the price of anarchy of the Pigouvian tax is much worse, $\Theta(H)$. This suggests that unless the social planner can influence the equilibrium selection process, the Pigouvian tax  might not be the right choice from a worst-case perspective.
	The quantitative version of our main theorem shows a tight lower bound on the taxed price of anarchy of $\LL_H$:
		
	\vspace*{0.5cm}
	\textbf{Theorem (``Lower Bound", quantitative version)}: 
	For $\LL_H$,  the {family of functions with bounded derivative ratio $H>1$, it holds that the taxed price of anarchy of $\LL_H$ is at least $\sqrt H$, i.e. $TPoA(\LL_{H})\ge\sqrt{H}$. 
	\vspace{0.5cm}
			
	Combining these two results we immediately derive that the taxed price of anarchy of $\LL_H$ is exactly $\sqrt H$, i.e. $TPoA(\LL_{H})=\sqrt{H}$.
	
	We view the main message of our results to be that the Pigouvian tax does not ensure an even approximately efficient use of the public resource, and even though no other taxation scheme is good either, in certain senses, a different taxation scheme may have an advantage over the Pigouvian tax. These results were shown in a model that is obviously idealized. The next subsection considers a variety of more realistic generalizations and shows that the main conclusions  hold for these generalizations as well.
	
	\subsection{Ride Sharing Games}
			
	In the second part of this paper, we look at a generalization of the basic model that captures transportation problems with private cars and ride-sharing options of various types. Here are some natural ride-sharing settings that the model captures:
			
	\begin{itemize}
		\item {\bf Bus}: The basic model studied above may be viewed as capturing the situation where a commuter may take the bus instead of a private car. The public good in this case is the road, and the taxes are tolls. A bus-ride causes only negligible increased congestion on the road compared to a car ride. However, bus-riders also incur the cost of congestion. Bus passengers do not pay tolls.
		\item {\bf Carpools with tolls}:
		In this model, $K$ passengers may share a car, where $K$ is some fixed constant {capturing the capacity of a shared car}. The load on the road when $q$ fraction of the population take a private car and the other $1-q$ fraction carpool is $q + (1-q)/K$. In this model we assume that carpoolers equally split the toll.
		\item {\bf Carpools without tolls}:
		This model is similar to the previous one, only that we assume that carpools are toll-free.
	\end{itemize}
			
	One may also think of various intermediate models, e.g. where carpoolers receive a discount on tolls. In Section \ref{sec:RSG-def} we define a parameterized common generalization of all of these models which we call \emph{Ride Sharing Games}. We prove analogs of our public-good congestion model results for this general class of settings. The main takeaway is that no toll function provides a good price of anarchy in any of these models and that from a worst-case perspective in all of these models there are toll functions that do better than the Pigouvian tax.
			
	In our models so far all users suffers from the externalities of the usage of the public good, i.e. of the road. One can alternatively consider a model in which the population that is not using the public good does not suffer from the externality imposed by public-good users. A simple example for that setting is a public road for which there is a parallel train track, and each person needs to decide whether to ride a private car or take the train. The train creates no congestion at all on the road, and train passengers do not suffer from the road delays at all (travel time is independent of number of cars on the road). We show that also for this model, no toll function provides a good price of anarchy.
			
	\vspace{0.2in}
	The rest of the paper is structured as follows: We first present related work in Section \ref{sec:RelatedWork}. Section \ref{sec:Game-def} formally defines the public-good congestion model. We phrase and prove our results for this model in Section \ref{sec:Results}.  Section \ref{sec:RSG-def} describes the ride sharing model and extend our results to this general model. Finally, in Section \ref{sec:train} we consider the case that the negative externality imposed by the use of the public good is only suffered by the public-good users, and present similar impossibility results for this model as well.
			
	\subsection{Related Work}
	\label{sec:RelatedWork}
	
	Marginal cost pricing for public goods was first proposed by Pigou \cite{Pigou1920}. 
	Samuelson \cite{Samuelson54, Samuelson55} has mathematically formulated the theory of public goods. Baumol and Oats \cite{Baumol88} and Laffont \cite{Laffont88} provides a good introduction to the topic as well as further references.
	
	The concept of price of anarchy analysis was introduced by Koutsoupias and Papadimitriou \cite{Koutsoupias-CP99}, and used by Roughgarden and Tardos \cite{RoughgardenT02} in the context of selfish routing games. Since then this type of analysis has drawn much research attention in routing  games as well as a variety of other games (see e.g. \cite{Roughgarden05, Nisan-R-T-V07} as well as references within). The complementary notion of price of stability was introduced in \cite{Anshelevich-AD-JK-ET08}.
						
	Traffic congestion problems were first formulated and analyzed in the 1950s by Wardrop \cite{Wardrop52} and Beckman et al.\cite{Beckmann-CN-CW55}. A general model in this area is the \emph{Congestion Game} model defined by Rosenthal \cite{Rosenthal72} in 1972. As mentioned above, much of the early work on price of anarchy considered congestion games.
	Some work has also been done regarding tolls in routing games, starting with \cite{Cole-YD-TM06}, where it was shown that when using taxes one can induce the optimal state as an equilibrium. Since then a lot of work has been done on tolls in congestion games \cite{Karakostas-SK04, Fleischer-KJ-MM04, Swamy07, Harks-GS-MS08, Bonifaci-MS-GS11,Bilo-CV19}.
		
	Our model is different from the literature on tolls in congestion games in several respects: first, our public-good model does not fall into the category of routing games, since in our model \emph{everyone} suffers the externalities resulting from the usage of the public good, even those who do not use it. Additionally, the crux of our model is that we treat the demand function as unknown, which is conceptually different from the implicit fixed demand assumption in congestion models.
			
	In transportation research, technologies advancements as well as the option to combine several aspects to combat congestion has been researched extensively, both empirically, e.g. \cite{varaiya, eliasson} and theoretically, e.g. \cite{cramton2017, cramton, hall, michales, Mehr}.
	
	\section{The Model}
	\label{sec:Game-def}
			
	A \emph{Public-Good Congestion Game (PGCG)} includes a population of agents where each of them may choose to use either a public good or a more costly alternative instead. Users of the public good ``congest” it, causing a negative externality to all others. We assume a large population of agents, with each agent by itself having an infinitesimal effect on the system. This is formalized by modeling the population as a continuum, who have a total mass $1$. We assume that the amount that each agent pays for the alternative might vary. We denote by $x\in[0,1]$ the fraction of the population that use the public good while $1-x$ is the fraction of the population that use the alternative.
	The externality of using the public good is a function of the total mass of population using it. The social planner can impose a tax on the agents that use the public good, in order to affect agents chosen strategy, aiming to minimize the social cost.

	Formally, a Public-Good Congestion Game $PGCG(l,\alpha,t)$ is defined by three functions --  an externality function $l$, a tax function $t$ and an individual cost function $\alpha$ as follows:
	\begin{definition}[Externality Function $l$]
		We use $l:[0,1]\longrightarrow\mathbb{R}_{\ge0}$ to denote an \emph{externality function}, with $l(x)$ being the externality, in monetary terms, experienced by every individual,  when the mass of public-good users is $x$.
		\begin{assumption}
			We assume that a externality function $l(\cdot)$ satisfies the following:
			\begin{itemize}
				\item $l$ is non-negative
				\item $l$ is non-decreasing
				\item $l$ is continuous and left differentiable with a left derivative function denoted by $l'(\cdot)$
			\end{itemize}
		\end{assumption}
	\end{definition}
	
	\begin{definition}[Tax Function $t$]
		We use $t:[0,1]\longrightarrow\mathbb{R}_{\ge0}$ to denote a \emph{tax function}, with $t(x)$ being the tax that every public-good user pays when the mass of public-good users is $x$.
	\end{definition}
	Note that we do not make any assumptions on $t(\cdot)$ -- it can be any function (not necessarily increasing or continuous). Also note that we define the tax as non-negative (which means the social planner does not pay agents to use the public good), as we show in Appendix \ref{appendix:Subsidy} that a social planner that is aiming to minimize the social cost would never choose to use subsidies.
			
	\begin{definition}[Individual Cost Function $\alpha$]
		We use $\alpha:[0,1]\longrightarrow\mathbb{R}_{\ge0}$ to denote an \emph{individual cost function} with $\alpha(x)$ being the disutility, in monetary terms, an individual that lies on the $x$ percentile of the population suffers by refraining from using the public good. I.e. for every $x\in [0,1]$,  $x$ fraction  of the population  have disutility at least $\alpha(x)$ for not using the public good (and using the alternative instead).  
		\begin{assumption}
			We assume that an individual cost function $\alpha(\cdot)$ satisfies the following:
			\begin{itemize}
				\item $\alpha$ is non-negative
				\item $\alpha$ is non-increasing
			\end{itemize}
		\end{assumption}
		As $\alpha(\cdot)$ is a monotone function, it is integrable and additionally, for every $x\in[0,1]$ it has a limit from the left as well as a limit from the right.  
		For a given $x$, we use $\alpha(x^{-})=\inf\{\alpha(t)\text{ }|\text{ }t<x\}$ to denote the limit from the left at $x$, and $\alpha(x^{+})=\sup\{\alpha(t)\text{ }|\text{ }t>x\}$ to denote the limit from the right at $x$.
	\end{definition}
			
	Note our model assumes the externality, individual cost and tax functions are all measured in the same units of money. We  assume agents have the same value for money, and the differences in preferences between different agents are reflected in their individual cost function.
	For convenience we will denote $PGCG(l,\alpha,t)$ by $G(l,\alpha,t)$.
			
	\begin{definition}[Personal Cost $c_{A}$]
		We use $c_{A}:[0,1]\times[0,1]\longrightarrow\mathbb{R}_{\ge0}$ to denote a \emph{personal cost function}, with $c_{A}(i,x)$ being the personal cost (or disutility) of an agent in the $i\in[0,1]$ percentile choosing action $A\in \{PG, ALT\}$ (using the public good or using the alternative), when $x$ fraction of the population use the public good.
		
		This means that the personal cost can be written as follows:
		\begin{gather*}
			c_{PG}(i,x) = l(x) + t(x) \quad\quad \text{when $i$ uses public good}\\
			c_{ALT}(i,x) = \alpha(i) + l(x) \quad\quad \text{ when $i$ uses the alternative}\\
		\end{gather*}
		We assume agents are rational: 
		each infinitesimal agent will take the action (use the public good or use the alternative) that has lower personal cost for the agents (breaking ties arbitrarily).
	\end{definition}
			
	It is now possible to define means of analyzing Public-Good Congestion Games, by defining equilibrium points, as well as socially optimal points.
	\begin{definition}[Equilibrium Point $\hat{x}$]
		Given $G(l,\alpha,t)$, $\hat{x}$ is an \emph{equilibrium point} if the following hold:
		\begin{gather*}
			\forall i\le\hat{x}: \quad c_{ALT}(i,\hat{x}) \ge c_{PG}(i,\hat{x})\\
			\forall i>\hat{x}: \quad c_{ALT}(i,\hat{x}) \le c_{PG}(i,\hat{x})\\
		\end{gather*}
		I.e. the agents that choose to use the public good are the $\hat{x}$-fraction of the agents whose cost of the alternative is highest.
		Note that in the case where $\alpha$ and $t$ are continuous, any internal equilibrium point $\hat{x}\in (0,1)$ must satisfy 
		$ c_{ALT}(i,\hat{x}) = c_{PG}(i,\hat{x})$. That is, $\hat{x}$ is such that the agents in the $\hat{x}$ percentile are indifferent between using the public good and the alternative.
				
		We next present a simple characterization of equilibrium points. For that characterization it is convenient to define $\alpha(0^{-})=+\infty$ and $\alpha(1^{+})=-\infty$.
		\begin{observation}
			\label{obs:EqCond}
			Given $G(l,\alpha,t)$, $\hat{x}$ is an equilibrium point if and only if 
			\begin{equation}
				\label{eqn:EqCond}
				\alpha(\hat{x}^{-}) \ge  t(\hat{x}) \ge \alpha(\hat{x}^{+})\\
			\end{equation}
		\end{observation}
		\begin{proof}
			By definition $\hat{x}$ is an equilibrium point if and only if  
			\begin{equation*}
				\forall i\le\hat{x}<j : \quad \alpha(i) + l(\hat{x}) \geq l(\hat{x}) + t(\hat{x})\geq \alpha(j) + l(\hat{x})\\
			\end{equation*}
			As $\alpha$ is non-increasing, it has limits from both left and right at $\hat{x}$, and thus by subtracting $l(\hat{x})$ and taking the limits as $i$ increases to $\hat{x}$ and as $j$ decreases to $\hat{x}$ we get that if the above holds then 
			\begin{equation*}
				\alpha(\hat{x}^{-}) \geq t(\hat{x})\geq \alpha(\hat{x}^{+})\\
			\end{equation*}
			which is simply Equation (\ref{eqn:EqCond}).
			
			To prove the other direction we observe that as $\alpha$ is non-increasing, 
			$\forall i\le\hat{x}<j$ it holds that $\alpha(i)\geq \alpha(\hat{x}^{-})\geq \alpha(\hat{x}^{+})\geq \alpha(j)$.
			Combining this with Equation (\ref{eqn:EqCond}) we get that 
			\begin{equation*}
				\forall i\le\hat{x}<j : \quad \alpha(i) + l(\hat{x}) \geq l(\hat{x}) + t(\hat{x})\geq \alpha(j) + l(\hat{x})\\
			\end{equation*}
			which is equivalent to $\hat{x}$ being an equilibrium point.
		\end{proof}
	\end{definition}
			
	Note that if $t$ is continuous then an equilibrium is guaranteed to exist.
	We denote by $EQ(l,\alpha,t)$ the set of equilibrium points of the game $G(l,\alpha,t)$. When $l,\alpha$ and $t$ are clear from context, we use $EQ$ to denote $EQ(l,\alpha,t)$.
			
	\begin{definition}[Social Cost $SC$]
		We use $SC_{(l,\alpha)}:[0,1]\longrightarrow\mathbb{R}_{\ge0}$ to denote the \emph{Social Cost} (or total disutility) function for the game  $G(l,\alpha,t)$, with $SC_{(l,\alpha)}(x)$ being the total amount of disutility when the $x$ fraction of the population with the highest individual cost values uses the public good in the game $G(l,\alpha,t)$.
		As agents are infinitesimal, the social cost can be given by integrating the personal cost over all $i\in[0,1]$, while omitting the taxes part in the cost, as taxes are paid to the social planner and do not affect the total social cost when considering the planner as part of society:
		\begin{equation}
			\label{eqn:SocialCost}
			SC_{(l,\alpha)}(x) = l(x) + \int_{x}^{1}\alpha(z)dz = l(0) + \int_{0}^{x}l'(z)dz + \int_{x}^{1}\alpha(z)dz
		\end{equation}
		
		When $l$ and $\alpha$ are clear from context we omit them in the notation and denote $SC_{(l,\alpha)}(x)$ by $SC(x)$.
	\end{definition}
	As $\alpha$ is integrable, the function $A(x)=\int_{x}^{1}\alpha(z)dz$ is continuous. $l$ is continuous as well, so the function $SC$ is continuous on the compact set $[0,1]$. This means that the minimum of the function $SC$ is obtained for some $x\in[0,1]$. We will call such a point a social optimal point:
			
	\begin{definition}[Social Optimal Point $x^{*}$]
		Given a game $G(l,\alpha,t)$, $x^{*}$ is called a \emph{social optimal point} if its social cost is minimal:
		\begin{equation}
			x^{*} \in {\arg\min}_{x\in[0,1]}SC_{(l,\alpha)}(x)
		\end{equation}
		Note that there might be multiple social optimal points, but every such point has the same minimal social cost. 
		With a slight abuse of notation we denote that 
		\emph{optimal social cost} by $SC_{(l,\alpha)}(x^{*})$.
	\end{definition}
			
	\begin{definition}[Price of Anarchy $PoA(l,\alpha,t)$]
		Given a game $G(l,\alpha,t)$,
		the \emph{Price of Anarchy $PoA(l,\alpha,t)$} is given by the largest ratio between the social cost of an equilibrium, and the optimal social cost\footnote{If the optimal social cost is zero then the price of anarchy is defined to be infinity, unless every equilibrium has zero social cost, in that case the PoA is defined to be 1. }:
		\begin{equation}
			PoA\left(l,\alpha,t\right) = \frac{\underset{\hat{x}\in EQ(l,\alpha,t)}{\sup}SC_{(l,\alpha)}(\hat{x})}{SC_{(l,\alpha)}(x^{*})}
		\end{equation}
		In the case that the set $EQ(l,\alpha,t)=\emptyset$ we define $PoA\left(l,\alpha,t\right)=1$.
	\end{definition}
			
	Given an externality function $l$, the social planner wishes to find a tax function $t=t_{l}$ (that may depend on $l$) that minimizes the social cost in the worst case over the population's individual cost function, ensuring a good outcome even if the population's individual cost function arbitrarily changes. Hence, we interpret the taxed price of anarchy as a bound ensuring the tax is good in the worst case (over preferences):
			
	\begin{definition}[Taxed Price of Anarchy $TPoA(l)$]
		The \emph{Taxed Price of Anarchy for externality function $l$}, $TPoA(l)$, is the price of anarchy of $l$ with the best tax function $t=t_{l}$ for the worst individual cost function $\alpha$:
		\begin{equation}
			TPoA(l) = \underset{t}{\inf}\text{  }\underset{\alpha}{\sup}\text{  } PoA\left(l,\alpha,t\right)
		\end{equation}
	\end{definition}
			
	\begin{definition}[Taxed Price of Anarchy $TPoA(\LL)$]
		The \emph{Taxed Price of Anarchy of a family (set) of externality functions $\LL=\{l(\cdot)\}$}, $TPoA(\LL)$, is the taxed price of anarchy of the worst externality function in the family:
		\begin{equation}
			TPoA(\LL) = \sup_{l \in \LL} TPoA(l)
		\end{equation}
	\end{definition}
			
	\section{Bounds on TPoA}
	\label{sec:Results}
	Pigou \cite{Pigou1920} has proven that the Pigouvian tax induces an efficient equilibrium. An immediate conclusion is that when the externality derivative is increasing (the externality function is convex), every equilibrium is efficient (since there is a unique equilibrium). For completeness we first present this result and its proof. After presenting this result we first extend it to externality functions that are approximately-monotone and then present our main result - a tight bound on the taxed price of anarchy as a function of the worst latency derivative ratio.
						
	\begin{theorem}[Pigou 1920 \cite{Pigou1920} - Increasing Latency Derivative]\label{thm:increasing}
		Fix $G(l,\alpha,t)$.
		If the externality function $l$ has a (left) derivative $l'$ that is non-decreasing, then $TPoA(l)=1$. Furthermore, for such $l$,
		setting the Pigouvian tax $t(x)=l'(x)$ yields $PoA(l,\alpha,t)=1$ for any individual cost function $\alpha$.
	\end{theorem}
			
	\begin{proof}
		Note that by way of definition of TPoA, it is enough to show that the PoA for a specific tax function $t$ is 1.
		Fix the tax function to be $t(x)=l'(x)$ for every $x$.
		For $t(x)=l'(x)$, we get from Equation (\ref{eqn:EqCond}) that $\hat{x}$ is an equilibrium if and only if 
		\begin{equation}\label{eqn:EqCond_thrm1}
			\alpha(\hat{x}^{-}) \ge l'(\hat{x}) \ge \alpha(\hat{x}^{+})
		\end{equation}
		
		Since $l'$ is a non-decreasing function while $\alpha$ is a non-increasing function, this means there exists a single equilibrium point (or interval $(x_L,x_H)$ for which $\alpha(x) = l'(x)$ for every $x\in (x_L,x_H)$). 
		In case $l'$ and $\alpha$ do not intersect, this point will be either $\hat{x}=0$ or $\hat{x}=1$.\\
		We now claim that every point $\hat{x}$ satisfying Equation (\ref{eqn:EqCond_thrm1}) is also a socially optimal point: 
		From Equation (\ref{eqn:SocialCost}) we get that $SC(x) = l(0) + \int_{0}^{x}l'(z)dz + \int_{x}^{1}\alpha'(z)dz$. Combining this with the fact that $\alpha$ is non-increasing and that $l'(x)$ is non-decreasing, we get that for every $x<\hat{x}$
		\begin{gather*}
			\alpha(x)\ge\alpha(\hat{x})\ge l'(x)\\
			\Downarrow\\
			SC(x)-SC(\hat{x}) = \int_{x}^{\hat{x}}(\alpha(z)-l'(z))dz \ge 0
		\end{gather*}
		and for every $x>\hat{x}$
		\begin{gather*}
			l'(x)\ge l'(\hat{x})\ge\alpha(x)\\
			\Downarrow\\
			SC(x)-SC(\hat{x}) = \int_{\hat{x}}^{x}(l'(z)-\alpha(z))dz \ge 0
		\end{gather*}
		meaning that the social cost of any $x\ne\hat{x}$ can only increase over that of $\hat{x}$. Hence, for the  optimal social cost $SC(x^{*})$ we have 
		\begin{equation*}
			SC(\hat{x})=SC(x^{*})
		\end{equation*}
		which in turn means 
		\begin{equation*}
			TPoA(l) = \frac{SC(\hat{x})}{SC(x^{*})}=1
		\end{equation*}
	\end{proof}
			
	\emph{Theorem \ref{thm:increasing}} requires the (left) derivative of the externality function to be (weakly) increasing, which is a strong assumption. 
	We next relax this assumption, allowing it to only be ``close to" an increasing function, 
	captured by the concept of \emph{$\gamma$-approximately increasing function}:
	\begin{definition}[$\gamma$-approximately increasing function]
		Let $\gamma > 1$. Given a function $h:[0,1]\longrightarrow\mathbb{R}_{>0}$ we say $h(\cdot)$ is a \emph{$\gamma$-approximately weakly-increasing function with $H(\cdot)$}, if $H:[0,1]\longrightarrow\mathbb{R}_{>0}$ is a non-decreasing function such that for every $x\in[0,1]$ it holds:
		\begin{equation}
			1 \le \frac{h(x)}{H(x)} \le \gamma
		\end{equation}
		Note that if $h(\cdot)$ is itself a non-decreasing function, it is $1-$approximately weakly increasing (with $h(\cdot)$ itself).
	\end{definition}
			
	\begin{theorem}[$\gamma$-approximately Increasing $l'$]\label{thm:bus-gamma}
		Fix $G(l,\alpha,t)$. If the (left) derivative of the externality function $l'$ is $\gamma$-approximately weakly increasing with an integrable function $L'$, then $TPoA(l)\le \gamma$. Furthermore, for such $l$, setting the tax function to be $t(x)=L'(x)$, yields $PoA(l,\alpha,t)\le\gamma$ for every inconvenience function $\alpha$.
	\end{theorem}
	
	\begin{proof}
		By way of definition of TPoA, it is enough to show that the PoA for a given tax function $t$ is $\gamma$. 
		Denote the integral function of $L'$ by $L$ and uniquely define $L$ by setting $L(0)=l(0)$. That is, $L(x) = l(0) + \int_{0}^{x} L'(z) dz$.
		Let us look at two games, one with externality function $l$ and the other with $L$, in both of which considering the tax $t(x)=L'(x)$ for every $x$, and the same $\alpha$.
		Denote by $x^{*}_{L}$ a social optimal point for $G_{L} = G(L,\alpha,L')$ and by $x^{*}_{l}$ a social optimal point for $G_{l} = G(l,\alpha,L')$.\\
		Since $x^{*}_{L}$ is an  optimal point in $G_{L}$, We know that
		\begin{equation*}
			SC_{G_{L}}(x^{*}_{L}) = L(0) + \int_{0}^{x^{*}_{L}}L'(z)dz + \int_{x^{*}_{L}}^{1}\alpha(z)dz \le L(0) + \int_{0}^{x^{*}_{l}}L'(z)dz + \int_{x^{*}_{l}}^{1}\alpha(z)dz = SC_{G_{L}}(x^{*}_{l})
		\end{equation*}
		Since $l$ is $\gamma$-approximately increasing with $L$, and since $\gamma\geq 1$, it holds that $l'(x)\le \gamma \cdot L'(x)$ for every $x\in [0,1]$,  and thus
		\begin{equation*}
			\begin{split}
				SC{_{G_{L}}}(x^{*}_{L}) &= L(0) + \int_{0}^{x^{*}_{L}}L'(z)dz + \int_{x^{*}_{L}}^{1}\alpha(z)dz \ge L(0) +  \frac{1}{\gamma}\int_{0}^{x^{*}_{L}}l'(z)dz + \int_{x^{*}_{L}}^{1}\alpha(z)dz \\
				&=l(0) + \frac{1}{\gamma}\int_{0}^{x^{*}_{L}}l'(z)dz + \int_{x^{*}_{L}}^{1}\alpha(z)dz  \ge \frac{1}{\gamma}\left(l(0) + \int_{0}^{x^{*}_{L}}l'(z)dz + \int_{x^{*}_{L}}^{1}\alpha(z)dz\right) = \frac{1}{\gamma}SC{_{G_{l}}}(x^{*}_{L})
			\end{split}
		\end{equation*}
		A $\gamma$-approximately increasing $l'$ with $L'$ also means that $ l'(x)\ge L'(x)$, and we get that:
		\begin{equation*}
			SC{_{G_{L}}}(x^{*}_{l}) = L(0) + \int_{0}^{x^{*}_{l}}L'(z)dz + \int_{x^{*}_{l}}^{1}\alpha(z)dz \le l(0) + \int_{0}^{x^{*}_{l}}l'(z)dz + \int_{x^{*}_{l}}^{1}\alpha(z)dz = SC{_{G_{l}}}(x^{*}_{l})
		\end{equation*}
		
		We now note that if $L'$ is non-decreasing, $x^{*}_{L}$ is an equilibrium point in $G_{l}$, and that if $G_{l}$ has more than one equilibrium point, each such point has the same social cost as $x^{*}_{L}$. 
		Combining it all together, we get that for $\gamma$-approximately latency function derivative $l'$:
		\begin{equation*}
			\forall\hat{x}\in EQ_{(l,\alpha,L')}:\quad SC_{G_{l}}(\hat{x}) = SC_{G_{l}}(x^{*}_{L}) \le \gamma\cdot SC_{G_{L}}(x^{*}_{L}) \le \gamma\cdot SC_{G_{L}}(x^{*}_{l}) \le \gamma\cdot SC_{G_{l}}(x^{*}_{l})
		\end{equation*}
		Meaning
		\begin{equation*}
			PoA(l,\alpha,L') = \gamma
		\end{equation*}
		which implies that 
		\begin{equation*}
			TPoA(l) \le \gamma
		\end{equation*}
	\end{proof}
			
	We now prove our main result -- two complementary theorems that tightly bound the TPoA when the ratio between the maximal and minimal value of the externality derivative is bounded.
			
	\begin{theorem}[Bounded Latency Derivative Upper Bound]\label{thm:bus-upper}
		Given $H>0$, let $l$ be an externality function with (left) derivative $l'$ satisfying $\frac{\sup_{0\le x\le1} l'(x)}{\inf_{0\le x\le1} l'(x)} \leq  H$.
		Then, $TPoA(l)\le \sqrt{H}$. 
		Moreover, for any such $l$, the guarantee is obtained by the constant tax $t(x) =\sqrt{H}\cdot\inf_{0\le x\le1}l'(x)$ for every $x\in [0,1]$.
	\end{theorem}
			
	\begin{proof}
		By way of definition of TPoA, it is enough to show that the PoA for a given toll function $t$ is no greater than $\sqrt{H}$. In order to do that, we will choose  the constant toll $t(x) =\sqrt{H}\cdot\inf_{0\le x\le1}l'(x)$ for any $x$. We fix some socially optimal point $x^*$ and prove the claim for any equilibrium point $\hat{x}$, depending which of the two is larger (note that as the case of $x^{*}=\hat{x}$ is trivial, we omit it from the proof):
		\subsubsection*{First Case: $x^{*}<\hat{x}$} In this case, it holds that:
		\begin{gather*}
			SC(x^{*}) = l(0) + \int_{0}^{x^{*}} l'(z)dz + \int_{x^{*}}^{\hat{x}} \alpha(z)dz + \int_{\hat{x}}^{1}\alpha(z)dz\\
			SC(\hat{x}) = l(0) + \int_{0}^{x^{*}} l'(z)dz + \int_{x^{*}}^{\hat{x}} l'(z)dz + \int_{\hat{x}}^{1}\alpha(z)dz\\
		\end{gather*}
		and using the fact that for every $x>y$ and for every $a>0$, it holds that $\frac{x+a}{y+a}\le\frac{x}{y}$, we get:
		\begin{equation*}
			\begin{split}
				PoA &= \frac{SC(\hat{x})}{SC(x^{*})} \le \frac{\int_{x^{*}}^{\hat{x}} l'(z)dz}{\int_{x^{*}}^{\hat{x}} \alpha(z)dz}
				\le \frac{\int_{x^{*}}^{\hat{x}} \sup_{0\le x\le1}l'(x)dz}{\int_{x^{*}}^{\hat{x}} \alpha(\hat{x}^{-})dx}
				\leq \frac{\sup_{0\le x\le1}l'(x)[\hat{x}-x^{*}]}{\alpha(\hat{x}^{-})[\hat{x}-x^{*}]}\\
				&=\frac{\sup_{0\le x\le1}l'(x)}{\alpha(\hat{x}^{-})}
				\le\frac{\sup_{0\le x\le1}l'(x)}{t(\hat{x})} =\frac{\sup_{0\le x\le1}l'(x)}{\sqrt{H}\cdot\inf_{0\le x\le1}l'(x)} \leq\frac{H}{\sqrt{H}}=\sqrt{H}\\
			\end{split}
		\end{equation*}
		
		by the fact that for any equilibrium $\hat{x}$ it holds that $\alpha(\hat{x}^{-}) \ge  t(\hat{x})$ by Equation (\ref{eqn:EqCond}).
		\subsubsection*{Second Case: $x^{*}>\hat{x}$} In this case, it holds that:
		\begin{gather*}
			SC(x^{*}) = l(0) + \int_{0}^{\hat{x}} l'(z)dz + \int_{\hat{x}}^{x^{*}} l'(z)dz + \int_{x^{*}}^{1}\alpha(z)dz\\
			SC(\hat{x}) = l(0) + \int_{0}^{\hat{x}} l'(z)dz + \int_{\hat{x}}^{x^{*}} \alpha(z)dz + \int_{x^{*}}^{1}\alpha(z)dz\\
		\end{gather*}
		and using the fact that for every $x>y$ and for every $a>0$, it holds that $\frac{x+a}{y+a}\le\frac{x}{y}$, we get:
		\begin{equation*}
			\begin{split}
				PoA &= \frac{SC(\hat{x})}{SC(x^{*})} \le \frac{\int_{\hat{x}}^{x^{*}} \alpha(z)dz}{\int_{\hat{x}}^{x^{*}} l'(z)dz} \le \frac{\int_{\hat{x}}^{x^{*}} \alpha(\hat{x}^{+})dx}{\int_{\hat{x}}^{x^{*}} \inf_{0\le x\le1}l'(x)dz}
				=\frac{\alpha(\hat{x}^{+})[x^{*}-\hat{x}]}{\inf_{0\le x\le1}l'(x)[x^{*}-\hat{x}]}\\ &=\frac{\alpha(\hat{x}^{+})}{\inf_{0\le x\le1}l'(x)}
				\le\frac{t(\hat{x})}{\inf_{0\le x \le 1}l'(x)} =\frac{\sqrt{H}\cdot\inf_{0\le x\le1}l'(x)}{\inf_{0\le x \le 1}l'(x)}\leq \sqrt{H}\\
			\end{split}
		\end{equation*}
		by the fact that for any equilibrium $\hat{x}$ it holds that $\alpha(\hat{x}^{+}) \le  t(\hat{x})$ by Equation (\ref{eqn:EqCond}).
	\end{proof}
			
	\begin{theorem}[Bounded Latency Derivative Lower Bound]\label{thm:bus-lower}
		For any $H>0$ there exists an externality function $l$ with (left) derivative $l'$ satisfying $\frac{\sup_{0\le x\le1} l'(x)}{\inf_{0\le x\le1} l'(x)}\le H$, for which $TPoA(l)\ge \sqrt{H}$.
	\end{theorem}
			
	\begin{proof}
		Given $H>1$, let us define $\Delta=\frac{1}{H}$ and choose some $\epsilon \ll 1$, and look at the following $l'(x)$:
		\begin{equation*}
			l'(x) = 
			\begin{cases}
				1 & \quad 0 \le x \le 2\epsilon\\
				\Delta & \quad 2\epsilon < x \le 1 \\
			\end{cases}
		\end{equation*}
		
		Clearly, $l'$ upholds the conditions of the theorem. Note that by defining $l'$ we define $l$ up to a constant. We define $l(0)=0$, which uniquely defines $l$.
		By way of definition of TPoA, it is enough to show that for every tax function $t$ there exist an individual cost function $\alpha$ for which the PoA is greater than $\sqrt{H}$. In order to do that, we will split the proof to two cases, building $\alpha$ according to the value of $t$ at the point $x=\epsilon$, as illustrated in Figure \ref{fig:LowerBound} and formulated bellow.
		
		\begin{figure}[htp]
			\centering
			\includegraphics[width=5in]{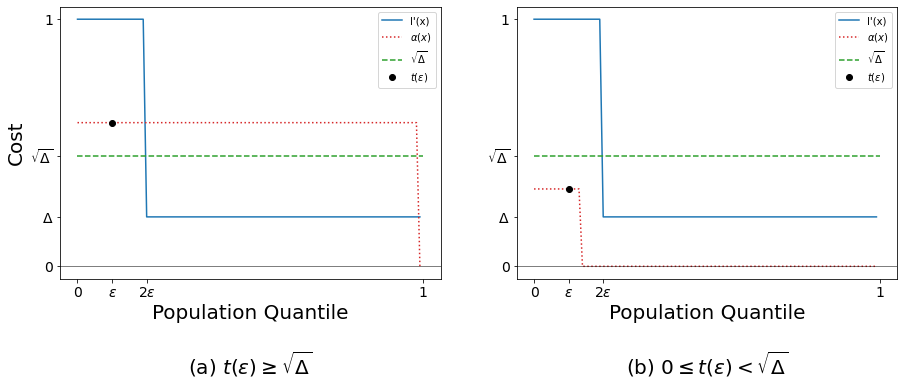}
			\caption{The externality function derivative $l'$ (a straight blue line) and the individual cost function $\alpha$ (a dashed red line), for the two cases of the proof: $t(\epsilon)\ge \sqrt{\Delta}$ in figure (a) and $0 \le t(\epsilon) < \sqrt{\Delta}$ in figure (b). In both cases, $x=\epsilon$ is an equilibrium point. While in (a) the social cost of the optimal point is at most $SC(1)$, in (b) the social cost of the optimal point is at most $SC(0)$.}
			\label{fig:LowerBound}
		\end{figure}	
		
		\subsubsection*{First Case: $t(\epsilon) \ge \sqrt{\Delta} $}
		If $t(\epsilon) \ge \sqrt{\Delta} $, we will define the following individual cost function $\alpha$:
		\begin{equation*}
			\alpha(x) = 
			\begin{cases}
				t(\epsilon) & \quad x< 1 \\
				0 & \quad x=1 \\
			\end{cases}
		\end{equation*}
		
		As $\alpha$ is continuous at $\epsilon$ and  $\alpha(\epsilon)= t(\epsilon)$, 
		the point $\hat{x}=\epsilon$ is an equilibrium (although not necessarily the only one). Using the fact that $l(x) = l(0) + \int_{0}^{x}l'(z)dz$, we get
		\begin{equation*}
			SC(x^{*})\le SC(1) = l(0) + \int_{0}^{1}l'(z)dz = 0 + 2\epsilon + \Big(1-2\epsilon\Big)\cdot \Delta = \Delta +2\epsilon\Big(1-\Delta\Big)< \Delta + 2\epsilon
		\end{equation*}
		as well as
		\begin{equation*}
			\begin{split}
				\max_{\hat{y}\in EQ}SC(\hat{y})&\ge SC(\hat{x})=SC(\epsilon) = l(\epsilon) + \int_{\epsilon}^{1} \alpha(z) dz = \epsilon + \Big(1-\epsilon\Big)\cdot t(\epsilon) \\
				&\ge \epsilon + \sqrt{\Delta}\Big(1-\epsilon\Big) = \sqrt{\Delta} + \epsilon \Big(1-\sqrt{\Delta}\Big)\geq  \sqrt{\Delta}  \\
			\end{split}
		\end{equation*}
		Thus, the taxed price of anarchy for such $l$ satisfies:
		\begin{equation*}
			TPoA(l)\ge \frac{SC(\epsilon)}{SC(1)} \geq  
			\frac{\sqrt{\Delta}}{\Delta+2\epsilon} \overset{\epsilon\rightarrow0}{\longrightarrow} \frac{1}{\sqrt{\Delta}} = \sqrt{H}
		\end{equation*}
		
		\subsubsection*{Second case: $0 \le t(\epsilon) < \sqrt{\Delta} $}
		If $0 \le t(\epsilon) < \sqrt{\Delta} $, then define $\alpha$ the following way:
		\begin{equation*}
			\alpha(x) = 
			\begin{cases}
				t(\epsilon) & \quad x\le \epsilon+\epsilon^{2} \\
				0 & \quad \epsilon+\epsilon^{2} < x \le 1 \\
			\end{cases}
		\end{equation*}
		By definition, the point $\hat{x}=\epsilon$ is an equilibrium (although not necessarily the only one). Using the fact that $l(x) = l(0) + \int_{0}^{x}l'(z)dz$, we get
		\begin{gather*}
			SC(x^{*})\le SC(0) = t(\epsilon)\cdot\Big(\epsilon+\epsilon^{2}\Big)\\
			\max_{\hat{y}\in EQ}SC(\hat{y})\ge SC(\hat{x}) = SC(\epsilon) = \epsilon +\epsilon^{2}\cdot t(\epsilon)\geq  \epsilon \\
		\end{gather*}
		
		and the taxed price of anarchy upholds:
		\begin{equation*}
			TPoA \ge \frac{SC(\epsilon)}{SC(0)} \geq \frac{\epsilon }{(\epsilon+\epsilon^{2})\cdot t(\epsilon)}
			= \frac{1}{(1+\epsilon)\cdot t(\epsilon)} = \frac{1}{1+\epsilon}\cdot\frac{1}{t(\epsilon)}\ge \frac{1}{1+\epsilon}\cdot\frac{1}{\sqrt{\Delta}} \overset{\epsilon\rightarrow0}{\longrightarrow} \frac{1}{\sqrt{\Delta}}=\sqrt{H}
		\end{equation*}
	\end{proof}
			
	\section{Ride Sharing Games}
	\label{sec:RSG-def}
	The Public-Good Congestion Game model studied so far captures the situation of congestion on a road, where the public good is the road usage. A passenger that chooses to take a private car uses the public good, but may alternatively take the bus. Taking the bus is less convenient, yet passengers that ride the bus do not increase the latency on the road. 
	In this section we expand the model to support additional ride sharing (carpooling) models in which passengers can ride in shared cars of limited capacity (unlike the bus).
			
	Shared car models can differ not only by their carpool capacity, but also by the tolls that carpooling passengers need to pay. One possibility is that carpools are charged the same as a private car and the toll is shared among its passengers, another is that the shared ride is exempt from toll. One can also consider intermediate cases in between (e.g., shared cars get 50\% discount on the toll which is still equally split).
	The general model we present captures all these models by adding two additional parameters to the game. The first is the marginal load a passenger in a shared ride adds to the road, and the other is the fraction of toll charged to a private car that a passenger sharing a ride is required to pay. 
	We next formally define this family of ride sharing games.
			
	\subsection{Ride Sharing Game Definition}
	A Ride Sharing Game $RSG(l,\alpha,t,\nu)$ is defined by four parameters.
	The first three parameters are the same as in the definition of Public-Good Congestion Games: $l(x)$ is the latency experienced by every road-taking agent when the mass of vehicles on the road is $x$, $\alpha(x)$ is the marginal inconvenience of the $x$'th percentile of population from sharing their ride with others and $t(x)$ is the toll that every private car taking the road needs to pay when the mass of vehicles on the road is $x$. The fourth parameter of a Ride Sharing Game is the \emph{Ride Sharing Technology} $\nu$ captured by a pair of parameters $\kappa$ and $\tau$:
	\begin{definition}[Ride Sharing Technology $\nu=\{\kappa,\tau\}$ ]
		We use $\nu=\{\kappa,\tau\}$ to denote the set of parameters defining the \emph{Ride Sharing Technology} used in the game:
		\begin{itemize}
			\item $\kappa\in[0,1)$ is the marginal load a passenger in a shared ride adds to the road, when normalizing the load of a private car to $1$.\footnote{The assumption that $\kappa<1$ corresponds to the externality of a passenger in a private car being larger than that of a passenger riding a shared car.} 
			Thus, with $x$ fraction taking private cars (and $1-x$ riding shared cars), the total load on the road is $x+\kappa(1-x)=(1-\kappa)x + \kappa$ which we denote by $\kappa(x)$. 
			For example, when ride sharing passengers do not add any additional load to the road (as the case for a bus) then $\kappa=0$, 
			while $\kappa=1/2$ means that a shared ride passenger add to the road half the load of a private car passenger.\footnote{This formulation is general enough to capture ride sharing vehicles that impose larger load on the road than private cars, as long as the per-passenger load is at most $1$: for example, a minibus with $20$ passengers that have total load of $2$ (like two private cars) has a per-passenger load of $1/10$.}
			\item $\tau\in[0,1)$ is the fraction of toll charged to a private car that a passenger sharing a ride is required to pay.\footnote{The assumption that $\tau<1$ corresponds to the natural assumption that the toll charged to a passenger of a private car is larger than the toll imposed on a passenger in a shared car.}\footnote{Note that the total toll paid by all carpool passengers might not necessarily be equal to the toll a private car is charged.} For example, $\tau=0$ means ride sharing passengers are exempt from paying any toll, while $\tau=1/2$ means a ride sharing passenger is  required to pay half the toll of a private car passenger.
		\end{itemize}
		For example, in the bus model, we consider the bus as a single vehicle that has negligible impact on the congestion of the road, and bus passengers do not add additional load to the road ($\kappa=0$). Buses are assumed to be exempt from tolls ($\tau=0$). Thus, using the above notations we denote the \emph{Bus Technology} as $BUS=\{0,0\}$. We present additional examples in Section \ref{sec:RSG-models}.
	\end{definition}
			
	Using these parameters, we redefine the personal cost of an agent.
			
	\begin{definition}[Personal Cost $c_{A}$]
		We use $c_{A}:[0,1]\times[0,1]\longrightarrow\mathbb{R}_{\ge0}$ to denote a \emph{personal cost function}, with $c_{A}(i,x)$ being the personal cost (or disutility) of a passenger in the $i\in[0,1]$ percentile choosing action $A\in \{CAR, RS\}$ (riding a private car or a ride sharing), when $x$ fraction of the population rides private cars.
		
		This means that when the ride sharing technology is $\nu=\{\kappa,\tau\}$ the personal cost can be written as follows:
		\begin{gather*}
			c_{CAR}(i,x) = l(\kappa(x)) + t(\kappa(x)) \quad\quad \text{when $i$ rides a private car}\\
			c_{RS}(i,x) = \alpha(i) + l(\kappa(x)) + \tau\cdot t(\kappa(x)) \quad\quad \text{ when $i$ rides a shared car}\\
		\end{gather*}
		We assume passengers are rational: 
		each infinitesimal passenger will take the action (ride a private car or a shared car) that has lower personal cost for the passenger (breaking ties arbitrarily).
	\end{definition}
			
	We will be interested in studying games induced by various RS technologies. For any RS technology $\nu$, we will consider the game for that fixed $\nu$ and denote $G_{\nu}(l,\alpha,t) = RSG(l,\alpha,t,\nu)$. The generalization of the game affects the expressions given for equilibrium points, as well as the social cost. The generalized equilibrium means a passengers in the equilibrium percentile $\hat{x}$ is not interested in changing his choice between riding a private car and sharing his ride, so the equilibrium condition becomes:
	\begin{equation}
		\alpha(\hat{x}^{-}) \ge  (1-\tau)\cdot t(\kappa(\hat{x})) \ge \alpha(\hat{x}^{+})\\
	\end{equation}
	and the social cost, which is given by integrating the personal cost over all $i\in[0,1]$, while omitting the tolls part in the cost, is given by
	\begin{equation}
		\label{eqn:GeneralSocialCost}
		SC_{(l,\alpha,\nu)}(x) = l(\kappa(x)) + \int_{x}^{1}\alpha(z)dz = l(\kappa)+\int_{0}^{x}l'(\kappa(z))(1-\kappa)dz + \int_{x}^{1}\alpha(z)dz
	\end{equation}
	The reader can find a full formal definition of the generalized game in Appendix \ref{appendix:GenModel}.
			
	\subsection{Examples of Ride Sharing Games}
	\label{sec:RSG-models}
	We next highlight several \emph{Ride Sharing Technologies} of interest:
	The bus model (which is equivalent to the \emph{Public-Good Congestion Game}), a carpool model where ride sharing passengers are required to pay the same toll as a private car (and split it equally), and a carpool model in which ride sharing passengers are exempt from paying the toll.
	We show how the parameters of the ride sharing technology capture these three models, and give explicit formulas for the personal and social costs of these scenarios.
			
	\subsubsection{Bus Model}
	In the bus model, an agent can choose whether to ride a private car or to ride a bus with infinite capacity sharing it with other agents being exempt from toll. In our model of an ``ideal bus'', the bus does not increase the congestion on the road, so every bus passenger has no impact on road congestion (meaning $\kappa=0$).\footnote{The assumption that no congestion is created by an infinite-size bus is ideal, and clearly not very realistic. A more nuanced model has finite-capacity ride sharing vehicles that do increase road congestion.}
	An agent in the $i$ percentile who chooses to take the bus pays their marginal inconvenience $\alpha(i)$ and is exempt from paying the toll on the road, as the bus itself is exempt from paying it (meaning $\tau=0$).\\
	We denote by $G_{BUS}(l,\alpha,t) = RSG(l,\alpha,t,\{0,0\})$ a bus model game in which the \emph{personal cost} of a passenger is
	\begin{gather*}
		c_{CAR}(i,x) =  l(x)+t(x) \quad\quad \text{when $i$ rides a private car}\\
		c_{BUS}(i,x) = \alpha(i) +l(x) \quad\quad \text{when $i$ rides the bus}\\
	\end{gather*}
	and the \emph{social cost $SC_{BUS}(x)$} is
	\begin{equation*}
		SC_{BUS}(x) = l(x) + \int_{x}^{1}\alpha(z)dz = l(0) + \int_{0}^{x}l'(z)dz + \int_{x}^{1}\alpha(z)dz
	\end{equation*}
	Meaning that the bus technology defined by the parameters $\nu=\{0,0\}$ induces the same game as the \emph{Public-Good Congestion Game}.
			
	\subsubsection{Finite Capacity Ride Sharing Models}
	In the finite-capacity ride sharing model, an agent can choose whether to ride a private car alone or to ride a shared car with additional $K-1$ agents (meaning $\kappa=\frac{1}{K}$) and thus the load on the road when $x$ fraction of the population ride private cars is $\kappa(x)=\frac{K-1}{K}x+\frac{1}{K}$. Though all agents are affected by the road latency, riding a shared car decreases the load.
	\begin{remark}
		As agents 'fill up' cars up to their capacity before using an additional shared car, the only 'shared' car that has less then the maximal capacity number of passengers is the last one, and as agents are infinitesimal, this is negligible.
	\end{remark}
	We highlight two finite-capacity ride sharing models: the \emph{Non-tolled Ride Sharing} model in which agents are exempt from paying the toll (meaning $\tau=0$) and the \emph{Tolled Ride Sharing} model in which the carpool is charged the same as a private car and its $K$ passengers split the toll equally (meaning $\tau=\frac{1}{K}$).
			
	\paragraph{Non-tolled Ride Sharing}
	We denote by $G_{RS}(l,\alpha,t)  = RSG(l,\alpha,t,\{\frac{1}{K},0\})$ a non-tolled ride sharing game in which the \emph{personal cost} of a passenger is:
	\begin{gather*}
		c_{CAR}(i,x) =   l\Big(\frac{K-1}{K}\cdot x + \frac{1}{K}\Big)+t\Big(\frac{K-1}{K}\cdot x + \frac{1}{K}\Big) \quad\quad \text{when $i$ rides a private car}\\
		c_{RS}(i,x) = \alpha(i) +  l\Big(\frac{K-1}{K}\cdot x + \frac{1}{K}\Big) \quad\quad \text{ when $i$ share a ride}\\
	\end{gather*}
			
	\paragraph{Tolled Ride Sharing}
	We denote by $G_{RS-TOLLED}(l,\alpha,t)  = RSG(l,\alpha,t,\{\frac{1}{K},\frac{1}{K}\})$ a tolled ride sharing game in which the \emph{personal cost} of a passenger is:
	\begin{gather*}
		c_{CAR}(i,x) =   l\Big(\frac{K-1}{K}\cdot x + \frac{1}{K}\Big)+t\Big(\frac{K-1}{K}\cdot x + \frac{1}{K}\Big) \quad\quad \text{when $i$ rides a private car}\\
		c_{RS-TOLLED}(i,x) = \alpha(i) +  l\Big(\frac{K-1}{K}\cdot x + \frac{1}{K}\Big) + \frac{1}{K}\cdot t\Big(\frac{K-1}{K}\cdot x + \frac{1}{K}\Big)\quad\quad \text{ when $i$ share a tolled ride}\\
	\end{gather*}
	
	Note that as the \emph{social cost $SC_{l,\alpha,\nu}(x)$} is not affected by tolls, the social cost of both these models is the same:
	\begin{equation*}
		\begin{split}
			SC_{RS}(x) = SC_{RS-TOLLED}(x) &= l\left(\frac{K-1}{K}\cdot x + \frac{1}{K}\right) +\int_{x}^{1}\alpha(z)dz\\
			&= l(\frac{1}{K}) + \int_{0}^{x} \frac{K-1}{K} \cdot l'\left(\frac{K-1}{K}\cdot z + \frac{1}{K}\right)dz + \int_{x}^{1}\alpha(z)dz\\
		\end{split}
	\end{equation*}
	
	Table \ref{table:PTTech} summarize the ride sharing technologies of these highlighted models.
			
	\begin{table}[H]
		\caption{The ride sharing technologies parameters of the highlighted models.}
		\label{table:PTTech}
		\begin{minipage}{\columnwidth}
			\begin{center}
				\begin{tabular}{c c c}
					\toprule
					Model & RS marginal & Passenger's share \\
					& road load ($\kappa$) & of car toll($\tau$) \\
					\hline
					Bus (of infinite capacity)\footnote{Note that one can view the bus model as the limit of the shared ride models when $K$ goes to infinity.} & 0 & 0 \\
					(Non-tolled) Shared Rides with carpool capacity $K$ & $\frac{1}{K}$ & 0 \\
					Tolled Shared Rides with carpool capacity $K$ & $\frac{1}{K}$ & $\frac{1}{K}$\\
					\bottomrule
				\end{tabular}
				\bigskip\centering
				\footnotesize
			\end{center}
		\end{minipage}
	\end{table}
			
	\subsection{Results for Ride Sharing Games}
	\label{sec:RSG-results}
			
	We next present our results for Ride Sharing Games, which generalize the results we have presented for Public-Good Congestion Games. We start by generalizing Theorem \ref{thm:increasing}, showing that if the latency derivative is non-decreasing then every equilibrium under the Pigouvian tax is  efficient. 
			
	\begin{restatable}{theorem}{thmIncDev}[Increasing Latency Derivative]
		Fix $G_{\nu}(l,\alpha,t)$. 
		If the latency function $l$ has a (left) derivative $l'$ that is non-decreasing, then $TPoA_{\nu}(l)=1$. Furthermore, for such $l$, setting the tax $t(x)=\frac{1-\kappa}{1-\tau}l'(x)$ yields $PoA_{\nu}(l,\alpha,t)=1$ for any inconvenience function $\alpha$.
		
		Table \ref{table:tollValues} specify the values of $t$ for the models highlighted in Section \ref{sec:RSG-models}.
		
		\begin{table}[H]
			\caption{The toll function that ensures $TPoA_{\nu}(l)=1$ whenever $l'$ is non decreasing, presented for the highlighted models.}
			\label{table:tollValues}
			\begin{minipage}{\columnwidth}
				\begin{center}
					\begin{tabular}{c c}
						\toprule
						Model & $t(x)$\\
						\hline
						Bus (of infinite capacity) & $l'(x)$ \\
						(Non-tolled) Shared Rides with carpool capacity  $K$ & $\frac{K-1}{K} l'(x)$ \\
						Tolled Shared Rides with carpool capacity $K$ & $l'(x)$ \\
						\bottomrule
					\end{tabular}
					\bigskip\centering
					\footnotesize
				\end{center}
			\end{minipage}
		\end{table}
	\end{restatable}
			
	We next generalize Theorem \ref{thm:bus-gamma}, showing that if the latency derivative is  $\gamma$-approximately weakly increasing  then the  toll price of anarchy is at most $\gamma$. 
	\begin{restatable}{theorem}{thmAppInc}[$\gamma$-approximately increasing $l'$]
		Fix $G_{\nu}(l,\alpha,t)$. If the (left) derivative of the externality function $l'$ is $\gamma$-approximately weakly increasing with an integrable function $L'$, then $TPoA_{\nu}(l)\le \gamma$. Furthermore, for such $l$, setting the tax function to be $t(x)=\frac{1-\kappa}{1-\tau}L'(x)$, yields $PoA(l,\alpha,t)\le\gamma$ for every inconvenience function $\alpha$.
	\end{restatable}
			
	Finally, we generalize Theorem \ref{thm:bus-upper} and Theorem \ref{thm:bus-lower}, that tightly bound the taxed price of anarchy when the ratio between the maximal and minimal value of the externality derivative is bounded.
			
	\begin{restatable}{theorem}{thmUppBnd}[Bounded Latency Derivative Upper Bound]\label{thm:gen-lower}
		Fix $\nu$. Given $H>0$, for latency $l$ assume that the (left) latency function derivative $l'$ satisfies $\frac{\sup_{0<x\le1} l'(\kappa(x))}{\inf_{0<x\le1} l'(\kappa(x))} \leq  H$.
		Then, $TPoA_{\nu}(l)\le \sqrt{H}$.
	\end{restatable}
			
	\begin{restatable}{theorem}{thmLowBnd}[Bounded Latency Derivative Lower Bound]
		Fix $\nu$. For any $H>0$ there exists a latency function $l$ with (left) derivative $l'$ satisfying $\frac{\sup_{0<x\le1} l'(\kappa(x))}{\inf_{0<x\le1} l'(\kappa(x))}\le H$, for which $TPoA_{\nu}(l)\ge\sqrt{H}$.
	\end{restatable}
			
	Full proofs for the generalized theorems can be found in Appendix \ref{appendix:GenResults}.
			
	\section{Externality on Public Good Users Only (Train)}
	\label{sec:train}
	
	In this section we consider settings in which the population that does not use the public good does not suffer the externality imposed by public-good users. 
	It will be convenient to think of this model as the train model. In this model passengers may take a train instead of a private car. 
	Train riders do not cause any congestion on the road, nor do they incur the congestion costs or pay any tolls. Under the same definitions for the latency, toll and individual cost functions as in our original 'bus' model, the personal cost in this model is
	\begin{gather*}
		c_{CAR}(i,x) = l(x) + t(x) \quad\quad \text{when $i$ rides a private car}\\
		c_{TRAIN}(i,x) = \alpha(i) \quad\quad \text{when $i$ rides the train}\\
	\end{gather*}
	Which means the equilibrium condition is
	\begin{equation}
		\alpha(\hat{x}^{-}) \ge  l(\hat{x}) + t(\hat{x}) \ge \alpha(\hat{x}^{+})\\
	\end{equation}
	and the social cost, which is given by integrating the personal cost over all $i\in[0,1]$, while omitting the tolls part in the cost, is given by
	\begin{equation}
		\label{eqn:GeneralSocialCost}
		SC_{(l,\alpha)}(x) = x\cdot l(x) + \int_{x}^{1}\alpha(z)dz
	\end{equation}
	
	We prove that even for this model the TPoA is unbounded:
			
	\begin{theorem}
		In the train model, for any $Z>1$ there exists a non-decreasing latency function $l$, for which $TPoA(l)\ge Z$.
	\end{theorem}
			
	\begin{proof}
		Given $Z>1$, let us define $\Delta=\frac{1}{Z}$ and choose some small enough $\epsilon>0$, 
		for which we define $\delta=\epsilon\Delta^{2}$ and look at the following $l'(x)$:
		\begin{equation*}
			l'(x) = 
			\begin{cases}
				0 & \quad 0 \le x \le  \epsilon - \delta \\
				\frac{\Delta}{\delta} & \quad \epsilon - \delta < x \leq \epsilon\\
				0 & \quad \epsilon < x \le 1 \\
			\end{cases}
		\end{equation*}
		
		Note that by defining $l'$ we define $l$ up to a constant. We define $l(0)=0$, which uniquely define $l$. Clearly, $l$ is a valid latency function.
		By way of definition of TPoA, it is enough to show that for every given toll function $t$, there exists an individual cost function $\alpha$ for which the PoA is greater than $H$. In order to do that, we will split the proof to two cases, building $\alpha$ according to the value of $t$ at the point $x=\epsilon$.
				
		\subsubsection*{First Case: $l(\epsilon) + t(\epsilon) \ge 1 $}
		If $l(\epsilon) + t(\epsilon) \ge 1 $, we will define the following individual cost function $\alpha$:
		\begin{equation*}
			\alpha(x) = 
			\begin{cases}
				l(\epsilon) + t(\epsilon) & \quad x< 1 \\
				0 & \quad x=1 \\
			\end{cases}
		\end{equation*}
		
		As $\alpha$ is continuous at $\epsilon$ and $\alpha(\epsilon)= l(\epsilon) + t(\epsilon)$, the point $\hat{x}=\epsilon$ is an equilibrium (although not necessarily the only one). Using the fact that $l(x) = l(0)+ \int_{0}^{x}l'(z)dz$, we get
		\begin{equation*}
			SC(x^{*})\le SC(1) = 1\cdot l(1) = l(0) + \int_{0}^{1}l'(z)dz = 0 + \int_{\epsilon-\delta}^{\epsilon}l'(z)dz = \delta\cdot \frac{\Delta}{\delta}=\Delta
		\end{equation*}
		as well as
		\begin{equation*}
			\begin{split}
				\max_{\hat{y}\in EQ}SC(\hat{y})&\ge SC(\hat{x})=SC(\epsilon) = \epsilon\cdot l(\epsilon) + \int_{\epsilon}^{1} \alpha(z) dz \\
				&= \epsilon\cdot\delta\cdot\frac{\Delta}{\delta} + (1-\epsilon)\cdot \Big(l(\epsilon)+t(\epsilon)\Big) \ge  \epsilon\Delta\ + (1-\epsilon)\cdot 1 \ge 1-\epsilon\\
			\end{split}
		\end{equation*}
		Thus, the taxed price of anarchy for $l$ satisfies:
		\begin{equation*}
			TPoA(l)\ge \frac{SC(\epsilon)}{SC(1)} \geq  
			\frac{1-\epsilon}{\Delta} \overset{\epsilon\rightarrow0}{\longrightarrow} \frac{1}{\Delta} = Z
		\end{equation*}
				
		\subsubsection*{Second case: $0 \le l(\epsilon) + t(\epsilon) < 1 $}
		If $0 \le l(\epsilon) + t(\epsilon) < 1 $, then define $\alpha$ the following way:
		\begin{equation*}
			\alpha(x) = 
			\begin{cases}
				t(\epsilon) + l(\epsilon) & \quad 0\le x\le \epsilon \\
				0 & \quad \epsilon < x \le 1 \\
			\end{cases}
		\end{equation*}
		By definition, the point $\hat{x}=\epsilon$ is an equilibrium (although not necessarily the only one). Using the fact that $l(x) = l(0)+ \int_{0}^{x}l'(z)dz$, we get
		\begin{gather*}
			SC(x^{*})\le SC(\epsilon-\delta) = (\epsilon-\delta)\cdot l(\epsilon-\delta) + \int_{\epsilon-\delta}^{1}\alpha(z)dz = (\epsilon-\delta)\cdot0 + \delta\cdot\Big(l(\epsilon)+t(\epsilon)\Big)\le \delta\\
			\max_{\hat{y}\in EQ}SC(\hat{y})\ge SC(\hat{x}) = SC(\epsilon) = \epsilon\cdot l(\epsilon) +\int_{\epsilon}^{1}\alpha(z)dz= \epsilon\cdot\Big(l(0)+\Delta\Big) + 0 =\epsilon\cdot\Delta
		\end{gather*}
		
		and thus the taxed price of anarchy for $l$ satisfies:
		\begin{equation*}
			TPoA(l) \ge \frac{SC(\epsilon)}{SC(\epsilon-\delta)} \geq \frac{\epsilon\Delta }{\delta} = \frac{\epsilon\Delta}{\epsilon\Delta^{2}} = \frac{1}{\Delta}= Z
		\end{equation*}
		
	\end{proof}

	\bibliographystyle{abbrvnat}
	\bibliography{references}
	\appendix
	\section{Simulations}
	\label{appendix:simulations}
	Using a simulation engine we've created, 
	we examine best response dynamics with the Pigouvian tax, and the equilibria it reaches. We simulate a large population of passengers in which in every round one passenger  is chosen at random and best responses to the current state. 
	Our simulations indicate that the dynamics might reach high cost equilibria, even when it does not start from them.
	
	Consider imposing the Pigouvian tax on the instance with the externality derivative
	\begin{equation*}
		\tilde{l}'(q) = \begin{cases}
			1 & \quad 0 \le x \le 0.12\\
			\Delta & \quad 0.12 < x \le 1\\
		\end{cases}
	\end{equation*}
	and the individual cost function
	\begin{equation*}
		\tilde{\alpha}(q) = \begin{cases}
			1.2 & \quad 0 \le x \le 0.06 \\
			0.8 & \quad 0.06 < x \le 1\\
		\end{cases}
	\end{equation*}
	
	With infinitesimal passengers  this instance has three equilibria points with the Pigouvian tax: $EQ =\{0.06, 0.12, 1\}$. Figure \ref{fig:SimulationsFunctions} illustrate the externality derivative, the individual cost function and the three Pigouvian equilibrium points.
	
	\begin{figure}[H]
		\centering
		\includegraphics[width=3.5in]{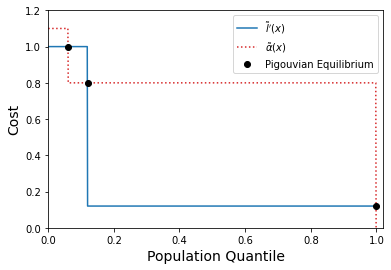}
		\caption{The externality derivative $\tilde{l}'$ and the individual cost function $\tilde{\alpha}$ used in the simulation, and the three Pigouvian equilibrium points they induce.}
		\label{fig:SimulationsFunctions}
	\end{figure} 
			
	We've created a simulation engine and have simulated a world that consists of a large number of agents (approximating a continuous population), each with the option to choose between using the public good and using the alternative. We randomly assign each of the agents an initial strategy and run a 'best-response' process in which we randomly choose an agent and update their strategy to best respond to the current state.
	All our simulations ended up with the dynamics converging to an equilibrium\footnote{We have checked that indeed each agent is best responding at the end of the process.}, and the convergence was fast. Depending on the initial number of agents starting using the public good, we see the process converges to either the optimal equilibrium point ($\hat{x}=1$) or the stable non-optimal equilibrium point ($\hat{x}=0.06$). Figure \ref{fig:Simulations} shows the convergence process, as well as the ratio between the social cost of the current state of the system and the optimal social cost. The graphs presented are from a single simulation run, however, all our executions produced similar results.
			
	\begin{figure}[H]
		\centering
		\includegraphics[width=5in]{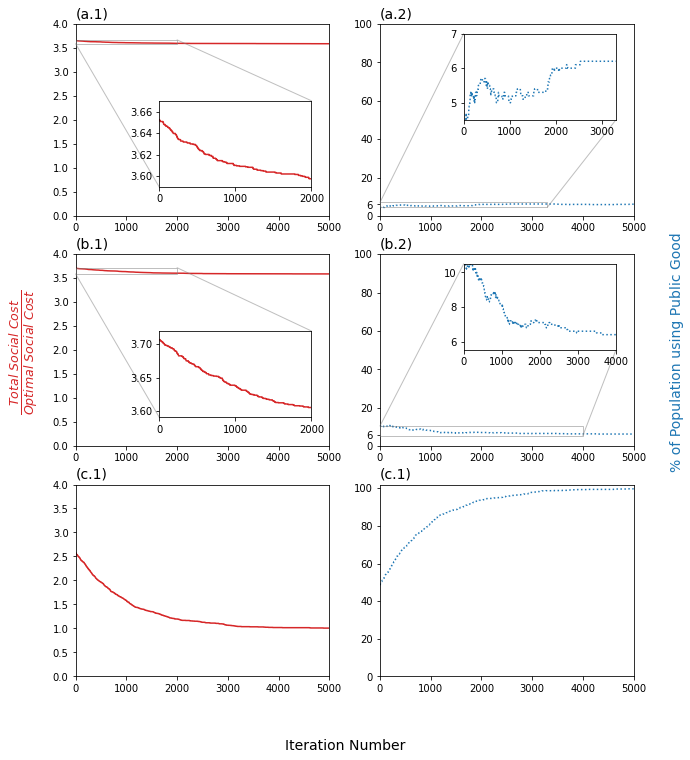}
		\caption{Three different simulations done for $\tilde{l}$ and $\tilde{\alpha}$ under the Pigouvian tax, showing different end results for a population of 1,000 agents. Figures (a.1) and (a.2) illustrates the process when $5\%$ ($x_{0}=0.05$) of the agents, chosen at random, are initially using the public good. Figure (b.1) and (b.2) illustrates the process when $10\%$ ($x_{0}=0.1$) of the agents, chosen at random, are initially using the public good. Figure (c.1) and (c.2) illustrates the process when $50\%$ ($x_{0}=0.5$) of the agents, chosen at random, are initially using the public good. While in the first two cases ($x_{0}=0.05$ and $x_{0}=0.1$) the 'best-reply' process converged to the equilibrium point $\hat{x}=0.06$ which is not optimal, in third one ($x_{0}=0.5$), the 'best-reply' process converged to the equilibrium point $\hat{x}=1$ which is optimal.}
		\label{fig:Simulations}
	\end{figure}
			
	Note that though the equilibrium point $\hat{x}=0.12$ exists, it is unstable, and even when we initialize the simulation to start in a point very close to it, the 'best-reply' process does not converge to it, as shown in Figure \ref{fig:SimulationUnstabelEq}.
			
	\begin{figure}[H]
		\centering
		\includegraphics[width=5in]{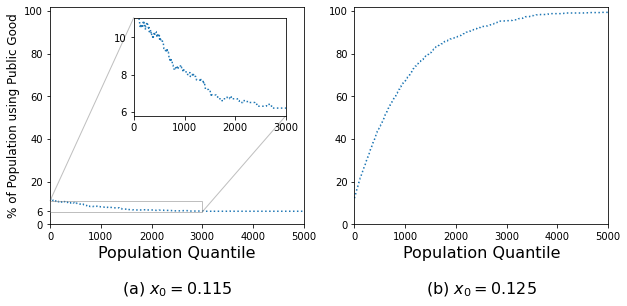}
		\caption{'Best-reply' process starting from $x_{0}=0.115$ (figure (a)) and $x_{0}=0.125$ (figure (b)). While the 'closest' equilibrium point for these two cases is $\hat{x}=0.12$, it is unstable and both scenarios converge to other, stable equilibrium points. While in the first case (figure (a)) the convergence is to $\hat{x}=0.06$, in the second case the convergence is to $\hat{x}=1$.}
		\label{fig:SimulationUnstabelEq}
	\end{figure}
			
	\section{Bounding the PoA for the Pigouvian Tax}
	\label{appendix:PigouBound}
	In this section we bound the price of anarchy for the Pigouvian tax showing it can be as large as $H$ for $\LL_H$. Formally, we prove the following claim:
	\begin{lemma}
		For a fixed constant $H>1$,  let $\LL_H$ be the family of externality functions that each satisfies  $l'(q_1)/l'(q_2) \le H$ for every $0 \le q_1, q_2 \le 1$. The price of anarchy of $\LL_H$ with the Pigouvian tax is $H$. 
	\end{lemma}
	
	\begin{proof}
		In order to prove this theorem, we prove an upper bound as well as a lower bound.
		\subsection*{Upper Bound}
		Given an externality function $l\in \LL_{H}$ and an individual cost function $\alpha$, we claim that $PoA \le H$. If the Pigouvian tax impose a single equilibrium point, then $PoA=1<H$ and we are done.
		For the case in which there is more than one equilibrium point, we will look at the worst equilibrium point $\hat{x}$ and the optimal equilibrium point $x^{*}$. We split to two cases:
		\subsubsection*{First Case : $x^{*}<\hat{x}$}
		In this case, it holds that:
		\begin{gather*}
			SC(x^{*}) = l(0) + \int_{0}^{x^{*}} \tilde{l}'(z)dz + \int_{x^{*}}^{\hat{x}} \tilde{\alpha}(z)dz + \int_{\hat{x}}^{1}\tilde{\alpha}(z)dz\\
			SC(\hat{x}) = l(0) + \int_{0}^{x^{*}} \tilde{l}'(z)dz + \int_{x^{*}}^{\hat{x}} \tilde{l}'(z)dz + \int_{\hat{x}}^{1}\tilde{\alpha}(z)dz\\
		\end{gather*}
		and using the fact that for every $x>y$ and for every $a>0$, it holds that $\frac{x+a}{y+a}\le\frac{x}{y}$, we get:
		\begin{equation*}
			\begin{split}
				PoA &= \frac{SC(\hat{x})}{SC(x^{*})} \le \frac{\int_{x^{*}}^{\hat{x}} \tilde{l}'(z)dz}{\int_{x^{*}}^{\hat{x}} \tilde{\alpha}(z)dz}
				\le \frac{\int_{x^{*}}^{\hat{x}} \sup_{0\le x\le1}\tilde{l}'(x)dz}{\int_{x^{*}}^{\hat{x}} \tilde{\alpha}(\hat{x}^{-})dx}
				\leq \frac{\sup_{0\le x\le1}\tilde{l}'(x)[\hat{x}-x^{*}]}{\tilde{\alpha}(\hat{x}^{-})[\hat{x}-x^{*}]}\\
				&=\frac{\sup_{0\le x\le1}\tilde{l}'(x)}{\tilde{\alpha}(\hat{x}^{-})}
				=\frac{\sup_{0\le x\le1}\tilde{l}'(x)}{t(\hat{x})} = \frac{\sup_{0\le x\le1}\tilde{l}'(x)}{\tilde{l}'(\hat{x})} \le
				\frac{\sup_{0\le x\le1}\tilde{l}'(x)}{\inf_{0\le x\le1}\tilde{l}'(x)} \le H\\
			\end{split}
		\end{equation*}
		
		\subsubsection*{Second Case: $x^{*}>\hat{x}$}
		In this case, it holds that:
		\begin{gather*}
			SC(x^{*}) = l(0) + \int_{0}^{\hat{x}} \tilde{l}'(z)dz + \int_{\hat{x}}^{x^{*}} \tilde{l}'(z)dz + \int_{x^{*}}^{1}\tilde{\alpha}(z)dz\\
			SC(\hat{x}) = l(0) + \int_{0}^{\hat{x}} \tilde{l}'(z)dz + \int_{\hat{x}}^{x^{*}} \tilde{\alpha}(z)dz + \int_{x^{*}}^{1}\tilde{\alpha}(z)dz\\
		\end{gather*}
		and using the fact that for every $x>y$ and for every $a>0$, it holds that $\frac{x+a}{y+a}\le\frac{x}{y}$, we get:
		\begin{equation*}
			\begin{split}
				PoA &= \frac{SC(\hat{x})}{SC(x^{*})} \le \frac{\int_{\hat{x}}^{x^{*}} \tilde{\alpha}(z)dz}{\int_{\hat{x}}^{x^{*}} \tilde{l}'(z)dz} \le \frac{\int_{\hat{x}^{+}}^{x^{*}} \tilde{\alpha}(\hat{x}^{+})dx}{\int_{\hat{x}}^{x^{*}} \inf_{0\le x\le1}l'(x)dz}
				=\frac{\tilde{\alpha}(\hat{x}^{+})[x^{*}-\hat{x}]}{\inf_{0\le x\le1}\tilde{l}'(x)[x^{*}-\hat{x}]}\\ &=\frac{\tilde{\alpha}(\hat{x}^{+})}{\inf_{0\le x\le1}\tilde{l}'(x)}
				\le\frac{t(\hat{x})}{\inf_{0\le x\le1}\tilde{l}'(x)} =\frac{\tilde{l}'(\hat{x})}{\inf_{0\le x\le1}\tilde{l}'(x)}\le\frac{\sup_{0\le x\le1}\tilde{l}'(x)}{\inf_{0\le x\le1}\tilde{l}'(x)}\leq H\\
			\end{split}
		\end{equation*}
		\subsection*{Lower Bound}
		We show an example for an externality function $l\in\LL_H$ and an individual cost function $\alpha$ for which the Pigouvian tax yields $PoA = H$. Given a fixed $H>1$, for a positive $\epsilon\ll 1$ consider the externality derivative
		\begin{equation*}
			\tilde{l}'(q) = \begin{cases}
				1 & \quad 0 \le x \le \epsilon\\
				\frac{1}{H} & \quad \epsilon < x \le 1\\
			\end{cases}
		\end{equation*}
		and the individual cost function
		\begin{equation*}
			\tilde{\alpha}(q) = \begin{cases}
				1 + \epsilon & \quad 0 \le x \le \frac{\epsilon}{2} \\
				1 - \epsilon & \quad \frac{\epsilon}{2} < x \le 1 - \epsilon\\
				0 & \quad 1 - \epsilon < x \le 1\\
			\end{cases}
		\end{equation*}
		Imposing the Pigouvian tax creates three possible equilibrium points: $EQ = \{\frac{\epsilon}{2},\epsilon,1-\epsilon\}$. The optimal equilibrium point is $\hat{x}=1-\epsilon$, for which the social cost is $SC(1-\epsilon) = \epsilon +(1-2\epsilon)\cdot \frac{1}{H}$. The worst equilibrium point is $\hat{x}=\epsilon$, for which the social cost is $SC(\epsilon) = \epsilon + (1-\epsilon)\cdot(1-2\epsilon)=1-2\epsilon+2\epsilon^{2}$. This means the price of anarchy is
		\begin{equation*}
			PoA = \frac{1-2\epsilon+2\epsilon^{2}}{\epsilon +(1-2\epsilon)\cdot \frac{1}{H}} = H \cdot \frac{1-2\epsilon+2\epsilon^{2}}{1-2\epsilon +\epsilon H}
		\end{equation*}
		Taking $\epsilon\longrightarrow0$, we get that $2\epsilon^2\ll\epsilon H$ and we get that
		\begin{equation*}
			PoA \overset{\epsilon\rightarrow0}{\longrightarrow} H
		\end{equation*}
	\end{proof}
			
	\section{Subsidy Inefficiency}
	\label{appendix:Subsidy}
	As we defined the tax function to be only positive, one could wonder if subsidy may help. while this is certainly not intuitive, in principle that might happen. In this section we show there is never a reason to choose a toll that is negative.
	\begin{lemma}
		\label{lemma:subsidy}
		Consider a Public-Good Congestion Game with externality function $l$. For every continuous tax $t:[0,1]\longrightarrow\mathbb{R}$ (that may use subsidies), there exists a non-subsidizing tax function $\tilde{t}:[0,1]\longrightarrow\mathbb{R}_{\ge0}$ such that for any $\alpha$, it holds that $PoA(l,\alpha,\tilde{t})\le PoA(l,\alpha,t)$.
	\end{lemma}
	\begin{proof}
		Given $l$ and $t$, consider the tax function that simply zeros all negative taxes of $t$:
		\begin{equation*}
			\tilde{t}(x) = 
			\begin{cases}
				0 & \quad x\ s.t.\ t(x) < 0 \\
				t(x) & \quad x\ s.t.\ t(x) \ge 0 \\
			\end{cases}
		\end{equation*}
		Note that if $t$ is continuous, $\tilde{t}$ is also continuous.
		
		Given $\bar{x}$ that is an equilibrium with $\tilde{t}$, we need to show that there exist $x$ that is an equilibrium with $t$ such that $SC_{l,\alpha}(x)\ge SC_{l,\alpha}(\bar{x})$.
		According to Equation \ref{eqn:EqCond}, $\bar{x}$ is an equilibrium point if and only if $\alpha(\bar{x}^{-})\ge t(\bar{x}) \ge \alpha(\bar{x}^{+})$. If $t(\bar{x})>0$ then $\bar{x}$ is also an equilibrium in $\tilde{t}$. So the only case we need to worry about is the one in which $t(\bar{x})\le0$. In this case, due to monotonicity and non-negativity of $\alpha$, for every $x \ge \bar{x}$ it holds that $\alpha(x)=0$ and since the externality function $l$ is non-decreasing, for every $x \ge \bar{x}$, we have that $SC(x)\ge SC(\bar{x})$.
		Hence, it is enough to show that for every such point $\bar{x}$ which is an equilibrium with $\tilde{t}$ but not with $t$, there exists some point $x\ge \bar{x}$ which is an equilibrium with $t$. Note that by defining $\alpha(1^+)=-\infty$, we cover the option that $\bar{x}=1$. This leaves us with the case for which $\bar{x}<1$. But if this is the case, and since $t$ is continuous, from the intermediate value theorem, one of the following two must hold: either there is some $x > \bar{x}$ for which $t(x)=0$, or for every $x \ge \bar{x}$, we have $t(x)\le 0$. In either case, we have a point $x \ge \bar{x}$ which is an equilibrium with $t$. As we've mentioned before, for that $x$ we have $SC(x)\ge SC(\bar{x})$, and we conclude our proof.
	\end{proof}
	
	Note that the proof does not work in the non-continuous case due to a technicality: when the tax function $t$ is non-continuous, there might not exist an equilibrium point and the problem is not well defined.
	
	We've proven the inefficiency of subsidies for Public-Good Congestion Games, but with minor changes, the proof can be adapted to prove the same result for all Ride Sharing Games.
	
	\section{Ride Sharing Model Formal Definition}
	\label{appendix:GenModel}
	For completion, we next present a full formal definition for the generalized \emph{Ride Sharing Game}. This section can be read as a stand alone text, and hence contain some duplication of Section \ref{sec:Game-def}.
	
	A \emph{Ride Sharing Game (RSG)} includes a population of passengers who all wishes to go from point $A$ to a point $B$ on a single road either by using a private car or by carpooling. We assume a large population of passengers, with each passenger by itself having an infinitesimal effect on the system. This is formalized by modeling the population as a continuum, who have a total mass $1$. 
	We assume that each passenger suffer some amount of inconvenience from choosing to share a ride over riding a private car, and that amount might differ between passengers. 
	We denote by $x\in[0,1]$ the fraction of the population that rides a private car while $1-x$ is the fraction of the population that carpool.
	The \emph{Ride Sharing (RS) technology} determines the mass of vehicles on the road as a function of the fraction of population choosing private cars, and the extent in which RS passengers suffer from road latency.
	The latency on the road is a function of the total mass of vehicles on the road. 
	The social planner can impose a toll on private cars taking the road, in order to affect passengers choice of transportation option, aiming to minimize the social cost. Toll is charged for every private car on the road, yet passengers of various ride sharing vehicles might be exempts of tolls.
	
	Formally, a Ride Sharing Game $RSG(l,\alpha,t,\nu)$ is defined by three functions -  a latency function $l$, a toll function $t$ and an inconvenience function $\alpha$, and by a ride sharing technology $\nu$ as follows:
	\begin{definition}[Latency Function $l$]
		We use $l:[0,1]\longrightarrow\mathbb{R}_{\ge0}$ to denote a \emph{latency function}, with $l(x)$ being the latency, in monetary terms, experienced by every agent when the mass of vehicles on the road is $x$.
		\begin{assumption}
			We assume that a latency function $l$ satisfies the following:
			\begin{itemize}
				\item $l$ is non-negative
				\item $l$ is non-decreasing
				\item $l$ is continuous and left differentiable with left derivative denoted by $l'$
			\end{itemize}
		\end{assumption}
	\end{definition}
	
	\begin{definition}[Toll Function $t$]
		We use $t:[0,1]\longrightarrow\mathbb{R}_{\ge0}$ to denote a \emph{toll function}, with $t(x)$ being the toll that every private car taking the road needs to pay when the mass of vehicles on the road is $x$.
	\end{definition}
	Note that we do not make any assumptions on $t$ -- it can be any function (not necessarily increasing or continuous). Also note that we define the toll as non-negative (which means the social planner does not pay agents to use the public good), as we've shown in Appendix \ref{appendix:Subsidy} that a social planner that is aiming to minimize the social cost would never choose to use subsidies.
	
	\begin{definition}[Inconvenience Function $\alpha$]
		We use $\alpha:[0,1]\longrightarrow\mathbb{R}_{\ge0}$ to denote an \emph{inconvenience function} with $\alpha(x)$ being the marginal inconvenience, in monetary terms, of the infinitesimal $x$ percentile of the population when carpooling.
		As passengers with higher dissatisfaction prefer private cars, it will be convenient to order the population according to their marginal inconvenience value from higher to lower.
		\begin{assumption}
			We assume that an inconvenience function $\alpha$ satisfies the following:
			\begin{itemize}
				\item $\alpha$ is non-negative
				\item $\alpha$ is non-increasing
			\end{itemize}
		\end{assumption}
		As $\alpha$ is a monotone function, it is integrable and additionally, for every $x\in[0,1]$ it has a limit from the left as well as a limit from the right.  
		For a given $x$, we use $\alpha(x^{-})=\inf\{\alpha(t)\text{ }|\text{ }t<x\}$ to denote the limit from the left at $x$, and $\alpha(x^{+})=\sup\{\alpha(t)\text{ }|\text{ }t>x\}$ to denote the limit from the right at $x$.
	\end{definition}
	
	Note our model assumes the latency, inconvenience and toll functions are all measured in the same units of money. We  assume agents have the same value for money, and the differences in preferences between different agents are reflected in the inconvenience function.
	
	We next define the ride sharing technology $\nu$, captured by the parameters $\kappa$ and $\tau$. 
	\begin{definition}[Ride Sharing Technology $\nu=\{\kappa,\tau\}$ ] 
		We use $\nu=\{\kappa,\tau\}$ to denote the set of parameters defining the \emph{ride sharing technology} used in the game:
		\begin{itemize}
			\item $\kappa\in[0,1)$ is the marginal load a passenger ride-sharing adds to the road, when normalizing the load of a private car to $1$.\footnote{The assumption that $\kappa<1$ corresponds to the externality of a passenger in a private car being larger than that of a passenger riding a shared car.} 
			Thus, with $x$ fraction taking private cars (and $1-x$ ride sharing), the total load on the road is $x+\kappa(1-x)=(1-\kappa)x + \kappa$ which we denote by $\kappa(x)$. 
			For example, when ride sharing passengers do not add any additional load to the road (as the case for a bus) then $\kappa=0$, 
			while $\kappa=1/2$ means that a ride sharing passenger add to the road half the load of a private car passenger.\footnote{This formulation is general enough to capture ride sharing vehicles that impose larger load on the road than private cars, as long as the per-passenger load is at most $1$: for example, a minibus with $20$ passengers that have total load of $2$ (like two private cars) has a per-passenger load of $1/10$. }
			\item $\tau\in[0,1)$ is the fraction of toll charged to a private car that a passenger riding a shared car is required to pay.\footnote{The assumption that $\tau<1$ corresponds to the toll charged to a passenger of a private car being larger than the tool imposed on a passenger that rides a shared car.}\footnote{Note that the total toll paid by all carpool passengers might not necessarily be equal to the toll a private car is charged.} For example, $\tau=0$ means ride sharing passengers are exempt from paying any toll, while $\tau=1/2$ means a ride sharing passenger is required to pay half the toll of a  private car passenger.
		\end{itemize}
	\end{definition}
	\begin{definition}[Personal Cost $c_{A}$]
		We use $c_{A}:[0,1]\times[0,1]\longrightarrow\mathbb{R}_{\ge0}$ to denote a \emph{personal cost function}, with $c_{A}(i,x)$ being the personal cost (or disutility) of a passenger in the $i\in[0,1]$ percentile choosing action $A\in \{CAR, RS\}$ (riding a private car or a shared car), when $x$ fraction of the population rides private cars.
		
		This means that when the ride sharing technology is $\nu=\{\kappa,\tau\}$ the personal cost can be written as follows:
		\begin{gather*}
			c_{CAR}(i,x) = l(\kappa(x)) + t(\kappa(x)) \quad\quad \text{when $i$ rides a private car}\\
			c_{RS}(i,x) = \alpha(i) + l(\kappa(x)) + \tau\cdot t(\kappa(x)) \quad\quad \text{ when $i$ share a ride}\\
		\end{gather*}
		We assume passengers are rational: 
		each infinitesimal passenger will take the action (ride a private car or share a ride) that has lower personal cost for the passenger (breaking ties arbitrarily).
	\end{definition}
	
	We will be interested in studying games induced by various  RS technologies. For any RS technology $\nu$, we will consider the game for that fixed $\nu$ and 
	denote $G_{\nu}(l,\alpha,t) = RSG(l,\alpha,t,\nu)$.
	It is now possible to define means of analyzing Ride Sharing games, by defining equilibrium points, as well as socially optimal points.
	\begin{definition}[Equilibrium Point $\hat{x}$]
		Given $G_{\nu}(l,\alpha,t)$, $\hat{x}$ is an \emph{equilibrium point} if the following hold:
		\begin{gather*}
			\forall i\le\hat{x}: \quad c_{RS}(i,\hat{x}) \ge c_{CAR}(i,\hat{x})\\
			\forall i>\hat{x}: \quad c_{RS}(i,\hat{x}) \le c_{CAR}(i,\hat{x})\\
		\end{gather*}
		I.e., the passengers that choose to use private cars are the $\hat{x}$-fraction of the passengers whose inconvenience from ride-sharing is highest.
		Note that in the case where $\alpha$ and $t$ are continuous, any internal equilibrium point $\hat{x}\in (0,1)$ must satisfy 
		$ c_{RS}(i,\hat{x}) = c_{CAR}(i,\hat{x})$, that is, $\hat{x}$ is such that the passenger in the $\hat{x}$ percentile is indifferent between riding a private car and sharing a ride.
		
		We next present a simple characterization of equilibrium points. 
		For that characterization it is convenient to define $\alpha(0^{-})=+\infty$ and $\alpha(1^{+})=-\infty$.
		\begin{lemma}
			Given $G_{\nu}(l,\alpha,t)$, $\hat{x}$ is an equilibrium point if and only if 
			\begin{equation}
				\label{eqn:GeneralEqCond}
				\alpha(\hat{x}^{-}) \ge  (1-\tau)\cdot t(\kappa(\hat{x})) \ge \alpha(\hat{x}^{+})\\
			\end{equation}
		\end{lemma}
		\begin{proof}
			By definition $\hat{x}$ is  a equilibrium point if and only if  
			\begin{equation*}
				\forall i\le\hat{x}<j : \quad \alpha(i) + l(\kappa(\hat{x})) + \tau\cdot t(\kappa(\hat{x})) \geq l(\kappa(\hat{x})) + t(\kappa(\hat{x}))\geq \alpha(j) + l(\kappa(\hat{x})) + \tau\cdot t(\kappa(\hat{x}))\\
			\end{equation*}
			As $\alpha$ is non-increasing, it has limits from both left and right at $\hat{x}$, and thus by 
			taking the limits as $i$ increases to $\hat{x}$ and as $j$ decreases to $\hat{x}$ we get that if the above holds then 
			\begin{equation*}
				\alpha(\hat{x}^{-}) + l(\kappa(\hat{x})) + \tau\cdot t(\kappa(\hat{x})) \geq l(\kappa(\hat{x})) + t(\kappa(\hat{x}))\geq \alpha(\hat{x}^{+}) +  l(\kappa(\hat{x})) + \tau\cdot t(\kappa(\hat{x}))\\
			\end{equation*}
			which is equivalent to Equation (\ref{eqn:GeneralEqCond}) by reorganizing. 
			
			To prove the other direction we observe that as $\alpha$ is non-increasing, 
			$\forall i\le\hat{x}<j$ it holds that $\alpha(i)\geq \alpha(\hat{x}^{-})\geq \alpha(\hat{x}^{+})\geq \alpha(j)$.
			Combining this with Equation (\ref{eqn:GeneralEqCond}) we get that 
			\begin{equation*}
				\forall i\le\hat{x}<j : \quad \alpha(i) +  l(\kappa(\hat{x})) + \tau\cdot t(\kappa(\hat{x})) \geq l(\kappa(\hat{x})) + t(\kappa(\hat{x}))\geq \alpha(j) + l(\kappa(\hat{x})) + \tau\cdot t(\kappa(\hat{x}))\\
			\end{equation*}
			which is equivalent to $\hat{x}$ being an equilibrium point.
		\end{proof}
		
		We denote by $EQ_{\nu}(l,\alpha,t)$ the set of equilibrium points of the game $G_{\nu}(l,\alpha,t)$. When $G_{\nu}(l,\alpha,t)$ is clear from context, we use $EQ$ to denote $EQ_{\nu}(l,\alpha,t)$.
	\end{definition}
	
	\begin{definition}[Social Cost $SC$]
		We use $SC_{(l,\alpha,\nu)}:[0,1]\longrightarrow\mathbb{R}_{\ge0}$ to denote the \emph{Social Cost} (or total disutility) function for the game  $G_{\nu}(l,\alpha,t)$, with $SC_{(l,\alpha,\nu)}(x)$ being the total amount of disutility when the $x$ fraction of the population with the highest marginal inconvenience values ride private cars in the game  $G_{\nu}(l,\alpha,t)$.
		As passengers are infinitesimal, the social cost can be given by integrating the personal cost over all $i\in[0,1]$, while omitting the tolls part in the cost, as tolls are paid to the social planner and do not affect the total social cost when considering the planner as part of society:
		\begin{equation}
			SC_{(l,\alpha,\nu)}(x) = l(\kappa(x)) + \int_{x}^{1}\alpha(z)dz = l(\kappa) + \int_{0}^{x}l'(\kappa(z))(1-\kappa)dz + \int_{x}^{1}\alpha(z)dz
		\end{equation}
		
		When $l$ and $\alpha$ are clear from context we omit them in the notation and denote  $SC_{(l,\alpha,\nu)}(x)$ by $SC_{\nu}(x)$.
	\end{definition}
	As $\alpha$ is integrable, the function $A(x)=\int_{x}^{1}\alpha(z)dz$ is continuous. $l$ is continuous as well, so the $SC$ is continuous on the compact set $[0,1]$. This means that the minimum of the function $SC$ is obtained for some $x\in[0,1]$. We will call such a point a social optimal point:
	
	\begin{definition}[Social Optimal Point $x^{*}$]
		Given a game $G_{\nu}(l,\alpha,t)$, $x^{*}$ is called a \emph{social optimal point} if its social cost is minimal:
		\begin{equation}
			x^{*} \in {\arg\min}_{x\in[0,1]}SC_{(l,\alpha,\nu)}(x)
		\end{equation}
		Note that there might be multiple social optimal points, but every such point has the same minimal social cost. 
		With a slight abuse of notation we denote that 
		\emph{optimal social cost} by $SC_{(l,\alpha,\nu)}(x^{*})$.
	\end{definition}
	
	\begin{definition}[Price of Anarchy $PoA_{\nu}(l,\alpha,t)$]
		Given a game $G_{\nu}(l,\alpha,t)$,	the \emph{Price of Anarchy $PoA_{\nu}(l,\alpha,t)$} is given by the largest ratio between the social cost of a Nash equilibrium, and the optimal social cost\footnote{If the optimal social cost is zero then the price of anarchy is defined to be infinity, unless every equilibrium has zero social cost, in that case the PoA is defined to be 1. }:
		\begin{equation}
			PoA_{\nu}\left(l,\alpha,t\right) = \frac{\underset{\hat{x}\in EQ_{\nu}(l,\alpha,t)}{\sup}SC_{(l,\alpha,\nu)}(\hat{x})}{SC_{(l,\alpha,\nu)}(x^{*})}
		\end{equation}
	\end{definition}
	Given  a ride sharing technology $\nu$ and a latency function $l$, the social planner wishes to find a toll function  $t=t_{l}$ (that may depend on $l$) that minimize the social cost in the worst case over the population's inconvenience function, ensuring a good outcome even if the population's preferences arbitrarily change. Hence, we interpret the tolled price of anarchy as a bound ensuring the toll is good in the worst case (over preferences):
	\begin{definition}[Tolled price of Anarchy $TPoA_{\nu}(l)$]
		Fix a ride sharing technology $\nu$. The \emph{Tolled Price of Anarchy $TPoA_{\nu}(l)$ for latency $l$} is the price of anarchy of $l$ with the best toll function $t=t_{l}$ for the worst inconvenience function $\alpha$:
		\begin{equation}
			TPoA_{\nu}(l) = \underset{t}{\inf}\text{  }\underset{\alpha}{\sup}\text{  } PoA_{\nu}\left(l,\alpha,t\right)
		\end{equation}
	\end{definition}
	\begin{definition}[Tolled Price of Anarchy $TPoA_{\nu}(\LL)$]
		Fix a ride sharing technology $\nu$. The \emph{Tolled Price of Anarchy of a family (set) of latency functions $\LL=\{l(\cdot)\}$}, $TPoA_{\nu}(\LL)$, is the tolled price of anarchy of the worst latency function in the family:
		\begin{equation}
			TPoA_{\nu}(\LL) = \sup_{l \in \LL} TPoA_{\nu}(l)
		\end{equation}
	\end{definition}

	\section{Ride Sharing Model Theorems}
	\label{appendix:GenResults}
	
	We restate and prove the ride sharing model theorems stated in Section \ref{sec:RSG-results}.
	
	\thmIncDev*
	
	\begin{proof}
		Note that by way of definition of TPoA, it is enough to show that the PoA for a specific tax function $t$ is 1.
		Fix the tax function to be $t(x)=\frac{1-\kappa}{1-\tau}l'(x)$ for every $x$.
		Note that this means
		\begin{equation*}
			(1-\tau)\cdot t(\kappa(x)) =
			(1-\tau)\cdot \frac{1-\kappa}{1-\tau}\cdot l'(\kappa(x))
			=  (1-\kappa)l'(\kappa(x))
		\end{equation*}
		Substitute this in Equation (\ref{eqn:GeneralEqCond}) and get that $\hat{x}$ is an equilibrium if and only if 
		\begin{equation}\label{eqn:GenEqCond_thrm1}
			\alpha(\hat{x}^{-}) \ge (1-\kappa)l'(\kappa(\hat{x})) \ge \alpha(\hat{x}^{+})
		\end{equation}
		
		Since $l'$ is a non-decreasing function and $\kappa(\cdot)$ is a linear function in $x$, $(1-\kappa)l'(\kappa(x))$ is a non-decreasing function while $\alpha$ is a non-increasing function, which means there exists a single equilibrium point (or interval $(x_L,x_H)$ for which $\alpha(x) = (1-\kappa)l'(\kappa(x))$ for every $x\in (x_L,x_H)$). 
		In case $(1-\kappa)l'(\kappa(x))$ and $\alpha(x)$ do not intersect, this point will be either $\hat{x}=0$ or $\hat{x}=1$.\\
		We now claim that every point $\hat{x}$ satisfying Equation (\ref{eqn:GenEqCond_thrm1}) is also a socially optimal point: 
		From Equation (\ref{eqn:GeneralSocialCost}) we get that $SC(x) = l(\kappa) + \int_{0}^{x}(1-\kappa)l'(\kappa(z))dz + \int_{x}^{1}\alpha'(z)dz$. Combining this with the fact that $\alpha$ is non-increasing and that $(1-\kappa)l'(\kappa(x))$ is non-decreasing, we get that for every $x<\hat{x}$
		\begin{gather*}
			\alpha(x)\ge\alpha(\hat{x})\ge (1-\kappa)l'(\kappa(x))\\
			\Downarrow\\
			SC(x)-SC(\hat{x}) = \int_{x}^{\hat{x}}(\alpha(z)-(1-\kappa)l'(\kappa(z)))dz \ge 0
		\end{gather*}
		and for every $x>\hat{x}$
		\begin{gather*}
			(1-\kappa)l'(\kappa(x))\ge (1-\kappa)l'(\kappa(\hat{x}))\ge\alpha(x)\\
			\Downarrow\\
			SC(x)-SC(\hat{x}) = \int_{\hat{x}}^{x}((1-\kappa)l'(\kappa(z))-\alpha(z))dz \ge 0
		\end{gather*}
		meaning that the social cost of any $x\ne\hat{x}$ can only increase over that of $\hat{x}$. Hence, for the  optimal social cost $SC(x^{*})$ we have 
		\begin{equation*}
			SC(\hat{x})=SC_{\nu}(x^{*})
		\end{equation*}
		which in turn means 
		\begin{equation*}
			TPoA(l) = \frac{SC(\hat{x})}{SC(x^{*})}=1
		\end{equation*}
	\end{proof}
	
	\thmAppInc*
	
	\begin{proof}
		By way of definition of TPoA, it is enough to show that the PoA for a given tax function $t$ is $\gamma$. 
		Denote the integral function of $L'$ by $L$ and uniquely define $L$ by setting $L(\kappa)=l(\kappa)$. That is, $L(\kappa(x))= l(\kappa) + \int_{0}^{x}(1-\kappa)L(\kappa(z))dz$. Let us look at two games, one with externality function $l$ and the other with $L$, in both of which considering the tax $t(x)=\frac{1-\kappa}{1-\tau}L'(x)$ for every $x$, and the same $\alpha(\cdot)$.
		Denote by $x^{*}_{L}$ a social optimal point for $G_{L} = G_{\nu}(L,\alpha,L')$ and by $x^{*}_{l}$ a social optimal point for $G_{l} = G_{\nu}(l,\alpha,L')$.\\
		Since $x^{*}_{L}$ is an  optimal point in $G_{L}$, We know that
		\begin{equation*}
			\begin{split}
				SC_{G_{L}}(x^{*}_{L}) &= L(\kappa) + \int_{0}^{x^{*}_{L}}(1-\kappa)L'(\kappa(z))dz + \int_{x^{*}_{L}}^{1}\alpha(z)dz \\
				&\le L(\kappa) + \int_{0}^{x^{*}_{l}}(1-\kappa)L'(\kappa(z))dz + \int_{x^{*}_{l}}^{1}\alpha(z)dz = SC_{G_{L}}(x^{*}_{l})
			\end{split}
		\end{equation*}
		Since $l$ is $\gamma$-approximately increasing with $L$, and since $\gamma\geq 1$, keeping in mind that $\kappa(x)=(1-\kappa)x + \kappa$, it holds that $l'(\kappa(x))\le \gamma \cdot L'(\kappa(x))$ for every $0 \le x \le 1$,  and thus
		\begin{equation*}
			\begin{split}
				SC{_{G_{L}}}(x^{*}_{L}) &= L(\kappa) + \int_{0}^{x^{*}_{L}}(1-\kappa)L'(\kappa(z))dz + \int_{x^{*}_{L}}^{1}\alpha(z)dz\\
				&\ge L(\kappa) + \frac{1}{\gamma}\int_{0}^{x^{*}_{L}}(1-\kappa)l'(\kappa(z))dz + \int_{x^{*}_{L}}^{1}\alpha(z)dz \\
				&= l(\kappa) + \frac{1}{\gamma}\int_{0}^{x^{*}_{L}}(1-\kappa)l'(\kappa(z))dz + \int_{x^{*}_{L}}^{1}\alpha(z)dz \\
				&\ge \frac{1}{\gamma}\left(l(\kappa) + \int_{0}^{x^{*}_{L}}(1-\kappa)l'(\kappa(z))dz + \int_{x^{*}_{L}}^{1}\alpha(z)dz\right) = \frac{1}{\gamma}SC{_{G_{l}}}(x^{*}_{L})\\
			\end{split}
		\end{equation*}
		A $\gamma$-approximately increasing $l'$ with $L'$ also means that $ l'(\kappa(x))\ge L'(\kappa(x))$, and we get that:
		\begin{equation*}
			SC{_{G_{L}}}(x^{*}_{l}) = L(\kappa) + \int_{0}^{x^{*}_{l}}(1-\kappa)L'(\kappa(z)dz + \int_{x^{*}_{l}}^{1}\alpha(z)dz \le l(\kappa) + \int_{0}^{x^{*}_{l}}(1-\kappa)l'(\kappa(z))dz + \int_{x^{*}_{l}}^{1}\alpha(z)dz = SC{_{G_{l}}}(x^{*}_{l})
		\end{equation*}
		
		We now note that if $L'$ is non-decreasing, $x^{*}_{L}$ is an equilibrium point in $G_{l}$, and that if $G_{l}$ has more than one equilibrium point, each such point has the same social cost as $x^{*}_{L}$. 
		Combining it all together, we get that for $\gamma$-approximately latency function derivative $l'$:
		\begin{equation*}
			\forall\hat{x}\in EQ_{(l,\alpha,L')}:\quad SC_{G_{l}}(\hat{x}) = SC_{G_{l}}(x^{*}_{L}) \le \gamma\cdot SC_{G_{L}}(x^{*}_{L}) \le \gamma\cdot SC_{G_{L}}(x^{*}_{l}) \le \gamma\cdot SC_{G_{l}}(x^{*}_{l})
		\end{equation*}
		Meaning
		\begin{equation*}
			PoA(l,\alpha,L') = \gamma
		\end{equation*}
		which implies that 
		\begin{equation*}
			TPoA(l) \le \gamma
		\end{equation*}
	\end{proof}
	
	\thmUppBnd*
	
	\begin{proof}
		By way of definition of TPoA, it is enough to show that the PoA for a given toll function $t$ is no greater than $\sqrt{H}$. let us look at $t(x)=\frac{1-\kappa}{1-\tau}\cdot\sqrt{H}\cdot\inf_{0<x\le1}l'(\kappa(x))$.
		We fix some socially optimal point $x^*$ and prove the claim for any equilibrium point $\hat{x}$, depending which of the two is larger (note that as the case of $x^{*}=\hat{x}$ is trivial, we omit it from the proof):
		
		\subsubsection*{First Case : $x^{*}<\hat{x}$}
		In this case, it holds that:
		\begin{gather*}
			SC(x^{*}) = l(\kappa) + \int_{0}^{x^{*}} (1-\kappa)l'(\kappa(z))dz + \int_{x^{*}}^{\hat{x}} \alpha(z)dz + \int_{\hat{x}}^{1}\alpha(z)dz\\
			SC(\hat{x}) = l(\kappa) + \int_{0}^{x^{*}} (1-\kappa)l'(\kappa(z))dz + \int_{x^{*}}^{\hat{x}} (1-\kappa)l'(\kappa(z))dz + \int_{\hat{x}}^{1}\alpha(z)dz\\
		\end{gather*}
		and using the fact that for every $x>y$ and for every $a>0$, it holds that $\frac{x+a}{y+a}\le\frac{x}{y}$, we get:
		\begin{equation*}
			\begin{split}
				PoA &= \frac{SC(\hat{x})}{SC(x^{*})} \le \frac{\int_{x^{*}}^{\hat{x}} (1-\kappa)\cdot l'(\kappa(z))dz}{\int_{x^{*}}^{\hat{x}} \alpha(z)dz}
				\le \frac{\int_{x^{*}}^{\hat{x}} (1-\kappa)\cdot\sup_{0<x\le1}l'(\kappa(x))dz}{\int_{x^{*}}^{\hat{x}} \alpha(\hat{x}^{-})dx} \\
				&\le\frac{(1-\kappa)\cdot\sup_{0<x\le1}l'(\kappa(x))[\hat{x}-x^{*}]}{\alpha(\hat{x}^{-})[\hat{x}-x^{*}]} = \frac{(1-\kappa)\cdot\sup_{0<x\le1}l'(\kappa(x))}{\alpha(\hat{x}^{-})}\\
				&\le\frac{(1-\kappa)\cdot\sup_{0<x\le1}l'(\kappa(x))}{[1-\tau]t(\kappa(\hat{x}))} =\frac{(1-\kappa)\cdot\sup_{0<x\le1}l'(\kappa(x))}{(1-\kappa)\cdot\sqrt{H}\cdot\inf_{0<x\le1}l'(\kappa(x))} =\frac{H}{\sqrt{H}}=\sqrt{H}\\
			\end{split}
		\end{equation*}
		
		by the fact that for any equilibrium $\hat{x}$ it holds that $\alpha(\hat{x}^{-}) \ge  [1-\tau]t(\kappa(\hat{x}))$ by Equation (\ref{eqn:GeneralEqCond}).
		
		\subsubsection*{Second Case: $x^{*}>\hat{x}$}
		In this case, it holds that:
		\begin{gather*}
			SC(x^{*}) = l(\kappa) + \int_{0}^{\hat{x}} (1-\kappa)l'(\kappa(z))dz + \int_{\hat{x}}^{x^{*}} (1-\kappa)l'(\kappa(z))dz + \int_{x^{*}}^{1}\alpha(z)dz\\
			SC(\hat{x}) = l(\kappa) + \int_{0}^{\hat{x}} (1-\kappa)l'(\kappa(z))dz + \int_{\hat{x}}^{x^{*}} \alpha(z)dz + \int_{x^{*}}^{1}\alpha(z)dz\\
		\end{gather*}
		and using the fact that for every $x>y$ and for every $a>0$, it holds that $\frac{x+a}{y+a}\le\frac{x}{y}$, we get:
		
		\begin{equation*}
			\begin{split}
				PoA = \frac{SC(\hat{x})}{SC(x^{*})} &\le \frac{\int_{\hat{x}}^{x^{*}} \alpha(z)dz}{\int_{\hat{x}}^{x^{*}} (1-\kappa)\cdot l'(\kappa(z))dz} \le \frac{\int_{\hat{x}}^{x^{*}} \alpha(\hat{x}^{+})dx}{\int_{\hat{x}}^{x^{*}} (1-\kappa)\cdot\inf_{0<x\le1}l'(\kappa(x))dz}
				\\
				&=\frac{\alpha(\hat{x}^{+})[x^{*}-\hat{x}]}{(1-\kappa)\cdot\inf_{0<x\le1}l'(\kappa(x))[x^{*}-\hat{x}]} =\frac{\alpha(\hat{x}^{+})}{(1-\kappa)\cdot\inf_{0<x\le1}l'(\kappa(x))}\\
				&=\frac{[1-\tau]t(\kappa(\hat{x}))}{(1-\kappa)\cdot\inf_{0<x\le1}l'(\kappa(x))} \le\frac{\sqrt{H}\cdot(1-\kappa)\cdot\inf_{0<x\le1}l'(\kappa(x))}{(1-\kappa)\cdot\inf_{0<x\le1}l'(\kappa(x))}\le\sqrt{H}\\
			\end{split}
		\end{equation*}
		by the fact that for any equilibrium $\hat{x}$ it holds that $\alpha(\hat{x}^{+}) \le [1-\tau]t(\kappa(\hat{x}))$ by Equation (\ref{eqn:GeneralEqCond}).
		
	\end{proof}
	
	\thmLowBnd*
	
	\begin{proof}
		Given $H>1$, let us define $\Delta=\frac{1}{H}$, choose some $\epsilon\ll 1$ and look at the following $l'(x)$:
		\begin{equation*}
			l'(x) = 
			\begin{cases}
				0 & \quad 0 \le x \le \kappa \\
				\frac{1}{1-\kappa} & \quad \kappa < x \le (1-\kappa)2\epsilon+\kappa \\
				\frac{1}{1-\kappa}\Delta & \quad (1-\kappa)2\epsilon+\kappa < x \le 1 \\
			\end{cases}
		\end{equation*}
		By manipulating the given $l'$, we get:
		\begin{equation*}
			(1-\kappa)l'(\kappa(x)) = 
			\begin{cases}
				1 & \quad 0 \le x \le 2\epsilon \\
				\Delta & \quad 2\epsilon < x  \le 1 \\
			\end{cases}
		\end{equation*}
		Clearly, $l'$ upholds the conditions of the theorem. Note that by defining $l'$ we define $l$ up to a constant. Let us choose $l(0)=0$ (meaning also $l(\kappa)=0$), which uniquely define $l$. Clearly, $l$ is a valid latency function.
		
		By way of definition of TPoA, it is enough to show that for every toll function $t$ there exist an individual cost function $\alpha$ for which the PoA is greater than $\sqrt{H}$. In order to do that, we will split the proof to two cases, building $\alpha$ according to the value of $t$ at the point $x=\epsilon$.
		
		\subsubsection*{First Case: $[1-\tau]t(\kappa(\epsilon)) \ge \sqrt{\Delta} $}
		If $[1-\tau]t(\kappa(\epsilon)) \ge \sqrt{\Delta} $, we will define the following marginal inconvenience function $\alpha$:
		\begin{equation*}
			\alpha(x) = 
			\begin{cases}
				[1-\tau]t(\kappa(\epsilon)) & \quad x< 1 \\
				0 & \quad x=1 \\
			\end{cases}
		\end{equation*}	
		
		As $\alpha$ is continuous at $\epsilon$ and  $\alpha(\epsilon)= [1-\tau]t(\kappa(\epsilon))$, the point $\hat{x}=\epsilon$ is an equilibrium (although not necessarily the only one). Using the fact that $l(\kappa)$ we can use $l(\kappa(x)) = \int_{0}^{x}(1-\kappa)l'(\kappa(z))dz$, and get
		\begin{equation*}
			SC_{\nu}(x^{*})\le SC_{\nu}(1) = 0 + 2\epsilon + \Big(1-2\epsilon\Big)\cdot\Delta = \Delta +2\epsilon\Big(1-\Delta\Big) < \Delta + 2 \epsilon
		\end{equation*}
		as well as
		\begin{equation*}
			\begin{split}
				\max_{\hat{y}\in EQ_{\nu}}SC_{\nu}(\hat{y})&\ge SC_{\nu}(\hat{x})=SC_{\nu}\left(\epsilon\right)\\
				&= l(\kappa(\epsilon)) + \int_{\epsilon}^{1} \alpha(z) dz = \epsilon + \left(1-\epsilon\right)\cdot [1-\tau]t(\kappa(\epsilon)) \\
				&\ge \epsilon + \left(1-\epsilon\right)\cdot \sqrt{\Delta} = \sqrt{\Delta} + \epsilon \left(1-\epsilon\right)\geq \sqrt{\Delta}\\			
			\end{split}
		\end{equation*}
		Thus, the tolled price of anarchy satisfies:
		\begin{equation*}
			TPoA_{\nu}(l)\ge \frac{SC(\epsilon)}{SC(1)} \geq  
			\frac{\sqrt{\Delta}}{\Delta+2\epsilon}\overset{\epsilon\rightarrow0}{\longrightarrow} \frac{1}{\sqrt{\Delta}} = \sqrt{H}
		\end{equation*}
		
		\subsubsection*{Second case: $0 \le [1-\tau]t(\kappa(\epsilon)) < \sqrt{\Delta} $}
		If $0 \le [1-\tau]t(\kappa(\epsilon)) < \sqrt{\Delta} $, then define $\alpha$ the following way:
		\begin{equation*}
			\alpha(x) = 
			\begin{cases}
				[1-\tau]t(\kappa(\epsilon)) & \quad x\le \epsilon + \epsilon^{2} \\
				0 & \quad \epsilon + \epsilon^{2} < x \le 1 \\
			\end{cases}
		\end{equation*}
		
		By definition, the point $\hat{x}=\epsilon$ is an equilibrium (although not necessarily the only one). Using the fact that $l(\kappa)$ we can use $l(\kappa(x)) = \int_{0}^{x}(1-\kappa)l'(\kappa(z))dz$ and get
		
		\begin{gather*}
			SC_{\nu}(x^{*})\le SC_{\nu}(0) = \Big(\epsilon+\epsilon^{2}\Big)\cdot [1-\tau]t(\kappa(\epsilon)) \le \Big(\epsilon+\epsilon^{2}\Big)\cdot \sqrt{\Delta} \\
			\max_{\hat{y}\in EQ}SC_{\nu}(\hat{y})\ge SC_{\nu}(\hat{x}) = SC_{\nu}\left(\epsilon\right) = \epsilon +\epsilon^{2}\cdot [1-\tau]t(\kappa(\epsilon)) \ge \epsilon\\
		\end{gather*}	
		
		and the tolled price of anarchy upholds:
		\begin{equation*}
			TPoA_{\nu}(l) \ge \frac{SC(\epsilon)}{SC(0)} \geq \frac{\epsilon}{(\epsilon+\epsilon^{2}) \cdot \sqrt{\Delta}}
			= \frac{1}{(1+\epsilon)\cdot t(\epsilon)}\geq \frac{1}{(1+\epsilon)\sqrt{\Delta}} \overset{\epsilon\rightarrow0}{\longrightarrow}  \frac{1}{\sqrt{\Delta}}=\sqrt{H}
		\end{equation*}
	\end{proof}
\end{document}